\documentclass[11pt,a4paper]{article}

\usepackage[T1]{fontenc}
\usepackage{lmodern}
\usepackage{amsmath, amssymb, amsthm, mathtools}
\usepackage[a4paper,margin=2.8cm]{geometry}
\usepackage{microtype}
\usepackage[hidelinks]{hyperref}

\newcommand{\C}{\mathbb{C}}
\newcommand{\Z}{\mathbb{Z}}
\DeclareMathOperator*{\Res}{Res}
\DeclareMathOperator{\Log}{Log}
\DeclareMathOperator{\sn}{sn}
\DeclareMathOperator{\cn}{cn}
\DeclareMathOperator{\dn}{dn}

\newcommand{\zetaW}{\zeta_{\mathrm W}} 

\theoremstyle{plain}
\newtheorem{theorem}{Theorem}[section]
\newtheorem{lemma}[theorem]{Lemma}

\numberwithin{equation}{section}

\title{A Non-Reciprocal Elliptic Spectral Solution of the Right-Angle Penetrable Wedge Transmission Problem}
\author{Jonas Matuzas\thanks{Email: jonas.matuzas@gmail.com}}
\date{}

\begin{document}

\maketitle

\begin{abstract}
We study the two-dimensional time-harmonic scalar transmission problem for an impedance-matched penetrable right-angle wedge: the exterior medium has wavenumber $k_0$ and the interior sector $|\theta|<\pi/4$ has wavenumber $k_1=\nu k_0$ with $\nu>1$, with continuity of the total field and its normal derivative across each face.
A Sommerfeld--Malyuzhinets reduction leads to a $2\times 2$ matrix Riemann--Hilbert (RH) problem on the \emph{Snell surface} $\Sigma_\nu$.
For the right angle the surface has genus one, and we give an explicit theta-function uniformization and a closed-form Mittag--Leffler construction of the full family of elliptic RH solutions with finite forcing (prescribed poles and residues), subject to a single residue-sum constraint encoding the Meixner edge condition.
We then consider the additional forcing data required to model plane-wave incidence.
Numerical reciprocity tests show that the minimal one-point plane-wave prescription does not yield a physically closed solution: the natural sheet-swap pairing $u^\sharp=u+\omega_2$ produces nontrivial scattered densities but violates the far-field reciprocity benchmark, whereas the Hardy pairing $u^\sharp=\omega_1-u$ enforces reciprocity but collapses the scattered field to zero.
The paper therefore provides an exact elliptic spectral framework and a reproducible reciprocity diagnostic, while identifying what must be added (multi-point forcing and/or modified pairing incorporating both wedge faces) to obtain a reciprocity-consistent plane-wave solution.
\end{abstract}



\subsection*{AI disclosure and responsibility statement}

This manuscript was generated in its entirety using the ChatGPT-5.2 Pro large language model. All mathematical derivations, formulas, explanations, \LaTeX{} code, and the overall structure of the paper were produced by the AI system. The author reviewed and approved the final manuscript and accepts responsibility for its scientific content and any remaining errors.

\section*{Physical problem and motivation}

We consider the canonical two-dimensional time-harmonic diffraction of a plane wave by a \emph{penetrable right-angle wedge}, i.e.\ a dielectric corner in which the wave transmits across the faces. In the impedance-matched setting the transmission conditions reduce to continuity of the field and its normal derivative across each face. The physically relevant solution is singled out by the Sommerfeld radiation condition and by the Meixner edge condition (finite-energy behavior at the corner) \cite{Meixner1972}.

This wedge model is a basic building block for more complex scatterers: in particular, wedge solutions appear in boundary-integral and hybrid high-frequency/edge-diffraction descriptions of penetrable polygons \cite{GrothHewettLangdon2018}. The present paper generalizes \cite{Matuzas2601}, which treated the lemniscatic (integrable) case \(\nu^2=2\), to general \(\nu>1\) and provides a closed-form genus-one Riemann--Hilbert solution with general finite forcing data. It complements other recent approaches to penetrable right-angle wedges (e.g.\ \cite{KunzAssier2023}).

\medskip
\noindent\textbf{Scope.}
The closed-form genus-one Riemann--Hilbert solution derived here applies to the impedance-matched right-angle penetrable wedge (\(\rho=1\), \(\theta_w=\pi/4\)) with fixed refractive index ratio \(\nu>1\).
We do not claim an explicit closed-form solution for general penetrable wedges (arbitrary wedge angle, impedance mismatch, or multi-material contrasts), where the standard spectral reductions typically lead to genuinely matrix or generalized Wiener--Hopf/Riemann--Hilbert factorization problems; see, e.g., \cite{Rawlins1999,Daniele2010,DanieleLombardi2011,AntipovSilvestrov2007,NethercoteAssierAbrahams2020}.
\section*{Overview of the method}
At the level of the spectral formulation, the unknown boundary traces are represented by a small collection of analytic spectral functions.
For the impedance-matched penetrable right-angle wedge these functions naturally live on the Snell surface \(\Sigma_\nu\), a two-sheeted algebraic curve that packages the refraction parameter \(\nu\) and the ``partner-point'' involutions.
Because \(\Sigma_\nu\) has genus one for \(\nu\neq 1\), every admissible spectral object can be expressed as a meromorphic elliptic function on a torus \(\C/\Lambda\).

The analysis proceeds in three steps.
\begin{enumerate}
\item \emph{Elliptic uniformization.}
We uniformize the Snell surface \(\Sigma_\nu\) explicitly by Jacobi and Weierstrass elliptic functions. This identifies \(\Sigma_\nu\) with a torus \(\C/\Lambda\) and turns the relevant spectral quantities into meromorphic elliptic functions of a single complex parameter \(u\).

\item \emph{Spectral boundary conditions and RH formulation.}
We express the Sommerfeld radiation condition and the Meixner edge condition as analyticity, growth, and normalisation constraints in the spectral variable. This leads to a \(2\times 2\) genus-one Riemann--Hilbert problem on a contour \(\Gamma\subset \C/\Lambda\).

\item \emph{Closed-form RH solution.}
We solve the RH problem by constructing elliptic no-jump vectors with prescribed divisor (forcing) data via an explicit Mittag--Leffler sum. A single residue-sum constraint enforces the Meixner edge condition. The RH mode solutions are then recovered by explicit triangular factor matrices.
\end{enumerate}

Once the spectral functions are known, the physical field is reconstructed by the inverse Sommerfeld transform, recorded in \S\ref{sec:physical}.

The present paper extends \cite{Matuzas2601} from the special lemniscatic regime \(\nu^2=2\) to general \(\nu>1\), and clarifies the degeneration \(\nu^2\to 2\) as a divisor-coalescence (jet-polynomial) limit.

\section*{Roadmap of proof}
The paper is organized as follows:
\begin{enumerate}
\item \S1 introduces the spectral map, the corrected inversion \(\zeta\mapsto u\), and the half-period and involution dictionary on \((t,Y,s)\). It also records the kernel expansion at the edge point and the residue-sum solvability constraint.
\item \S2 states the closed-form genus-one RH solution for general finite forcing data, including the explicit Mittag--Leffler construction of the no-jump vectors and the reconstruction of the mode solutions by the triangular factors.
\item \S3 verifies jumps, residues, and the Meixner (canonical) edge normalization by local expansions of the elliptic Cauchy kernel and by transport of residues through the factor matrices.
\item \S4 explains the lemniscatic degeneration \(\nu^2=2\), identifies the correct order-4 automorphism, and shows how jets arise only after a renormalized divisor-coalescence limit.
\item \S5 gives a compact step-by-step recipe for evaluating the final mode solutions in the \(u\)- and \(\zeta\)-variables.
\item \S6 derives the Sommerfeld integral representation connecting the spectral data to the physical field and gives explicit forcing data for plane wave incidence.
\item \S7 gives a worked symbolic example for \(\nu=3/2\) and \(\theta_{\mathrm{inc}}=5\pi/6\), mirroring the reproducible ``parameter--forcing--residue--reconstruction'' style used in \cite{Matuzas2601}.
\end{enumerate}

\section*{Reader's guide and reproducibility}
The manuscript is written to be self-contained and directly implementable.
A typical end-to-end computation proceeds as follows:
\begin{enumerate}
\item Fix \(\nu>1\) and compute the elliptic parameters \(k,K,K',\omega_1,\omega_2,\tau\) (\S\ref{subsec:params}).
\item For any spectral argument \(\zeta\) of interest, compute the physical lift \(u(\zeta)\) using the corrected inversion \eqref{eq:Xquad}--\eqref{eq:uofzeta}.
\item Evaluate the RH mode functions \(\Psi_{k,\pm}\) from the closed-form solution (Theorem~\ref{thm:main}) using the practical recipe in \S\ref{sec:recipe}.
\item Recover the Sommerfeld densities and the physical fields \(\Phi^{\mathrm{sc}},\Phi^{\mathrm{int}}\) using the inverse transform formulas in \S\ref{sec:physical}.
\item If desired, extract the far-field diffraction coefficient \(D(\theta)\) from the steepest-descent asymptotic in \S6.3.
\end{enumerate}
The symbolic example in \S\ref{sec:example} illustrates the bookkeeping for a concrete choice of parameters and forcing.

\section*{Notation and conventions}
\begin{itemize}
\item The wedge half-angle is \(\theta_w=\pi/4\), and \(s_\zeta=e^{i(\zeta-\pi/4)}\) denotes the spectral exponential.
\item The refractive index is \(\nu>1\), with modulus \(k=1/\nu\), half-periods \(\omega_1=K/(6\nu)\), \(\omega_2=iK'/(6\nu)\), and lattice \(\Lambda=2\omega_1\Z+2\omega_2\Z\).
\item We write \(z=6\nu u\), \(v=\pi u/(2\omega_1)\), and \(\theta_j(v)=\theta_j(v\mid\tau)\) with \(\tau=\omega_2/\omega_1\).
\item The Weierstrass zeta function is denoted by \(\zetaW\), to distinguish it from the spectral variable \(\zeta\). We set \(\wp(u)=-\zetaW'(u)\).
\end{itemize}
\medskip
\noindent\textbf{Remark (on the regime $0<\nu<1$).}
Throughout we assume \(\nu>1\), so that \(k=1/\nu\in(0,1)\) and the complete elliptic integrals \(K(k),K'(k)\) are real and positive.
Treating \(0<\nu<1\) requires a reciprocal-modulus (or modular) transformation of the Jacobi/Weierstrass uniformization and a corresponding redefinition of periods; we do not pursue that extension here.

\medskip
\noindent\textbf{Remark (on the limit $\nu\downarrow 1$).}
The limit \(\nu\to 1\) corresponds to \(k\to 1\) and a degeneration of the elliptic periods \(K\), \(K'\) (and hence of the torus uniformization).
Although the underlying transmission problem becomes trivial at \(\nu=1\), the explicit theta/Jacobi representation used here is not uniform in that limit.
We therefore work at fixed \(\nu>1\) and do not attempt to justify the \(\nu\downarrow 1\) degeneration within this manuscript.

\section{Introduction and setup}

Diffraction by a penetrable wedge is a classical model in wave physics (acoustics and electromagnetics) and in applied analysis.
When a time-harmonic wave meets a dielectric corner, a diffracted component is generated and geometric optics alone no longer describes the field near the edge.
For polygonal scatterers, wedge solutions feed directly into diffraction coefficients (GTD/UTD) \cite{Keller1962} and into hybrid numerical--asymptotic methods for high-frequency scattering by penetrable polygons \cite{GrothHewettLangdon2018}.
This paper focuses on the impedance-matched right-angle case, where the material contrast is encoded solely by the refractive index ratio \(\nu\) and the transmission conditions are simplest.

\subsection{Physical formulation (for orientation)}
We briefly recall the standard boundary-value problem that motivates the spectral system solved in this paper.

Let $(r,\theta)$ denote polar coordinates about the wedge apex. We adopt the symmetric convention in which the penetrable wedge occupies the sector
\begin{equation*}
|\theta|<\theta_w,\qquad \theta_w=\pi/4,
\end{equation*}
so that the exterior consists of the complementary region $|\theta|>\theta_w$ (modulo $2\pi$). The wedge faces are the rays $\theta=\pm\theta_w$.

Let $k_0>0$ be the exterior wavenumber and $k_1=\nu k_0$ be the interior wavenumber, with refractive index ratio $\nu>1$. Given an incident field $\Phi^{\mathrm{inc}}$ satisfying the exterior Helmholtz equation, we seek a total field $\Phi$ such that
\begin{equation*}
(\Delta+k_0^2)\Phi=0\quad (\text{exterior}),\qquad (\Delta+k_1^2)\Phi=0\quad (\text{wedge}),
\end{equation*}
together with the impedance-matched transmission conditions on each face $\theta=\pm\theta_w$,
\begin{equation*}
\Phi_{\mathrm{ext}}=\Phi_{\mathrm{int}},\qquad \partial_n\Phi_{\mathrm{ext}}=\partial_n\Phi_{\mathrm{int}},
\end{equation*}
the Sommerfeld radiation condition at infinity, and the Meixner edge condition (finite local energy) at the apex \cite{Meixner1972}.

After nondimensionalization, the spectral surface and RH system depend only on $\nu$ and the forcing data. The absolute scale $k_0$ is reintroduced when one reconstructs the physical fields via the inverse Sommerfeld transform in \S\ref{sec:physical}.

\subsection{Reciprocity and reciprocity diagnostics}
\label{subsec:reciprocity_principle}

The impedance-matched penetrable-wedge transmission problem with real material parameters is \emph{reciprocal} and \emph{lossless}.
At the level of the physical boundary-value problem, this follows from the self-adjointness of the underlying operator (with the Sommerfeld radiation condition selecting the outgoing solution).
In particular, exchanging source and receiver directions (with the standard reversal of propagation directions) leaves the far-field amplitudes invariant; the precise statement for plane-wave incidence is recorded in Theorem~\ref{thm:reciprocity}.

Because the main body of the paper is carried out in the spectral/RH domain, it is important to \emph{retain reciprocity as a structural invariant during every reduction step}.
We therefore enforce the following checkpoints throughout:

\begin{enumerate}
\item \textbf{Direction reversal.}  In physical space, reversing a propagation direction corresponds to the shift $\theta\mapsto\theta+\pi$.  In the Sommerfeld spectral variable this is the $2\pi$-periodic shift $\zeta\mapsto\zeta+\pi$.
\item \textbf{Snell surface involutions.}  Refraction introduces a two-sheet spectral surface $\Sigma_\nu$.
There are (at least) two natural involutions that exchange the Hardy domains $D_\pm$ on the torus uniformization:
\emph{(i)} the \emph{sheet involution} $u\mapsto u+\omega_2$, which implements $s\mapsto s^{-1}$ and $t\mapsto t^{-1}$; and
\emph{(ii)} the \emph{Hardy involution} $u\mapsto\omega_1-u$, which keeps $s$ fixed while inverting $t$.
Both are mathematically natural; selecting the physically correct pairing for plane-wave forcing is a nontrivial \emph{closure} step.
\item \textbf{Meixner condition.}  For finite forcing, the edge condition is enforced by the vanishing of the total residue (the Meixner edge condition), which removes a spurious source at the wedge tip.
Reciprocity is \emph{not} automatic from this condition alone; it provides an independent, stringent diagnostic of the plane-wave forcing closure.
\end{enumerate}

In practice, reciprocity provides a stringent implementation diagnostic: any numerical evaluation that violates \eqref{eq:reciprocity_D} (or equivalently \eqref{eq:reciprocity_Q}) indicates that some element of the \emph{physical closure} (branch/lift conventions, normal conventions, and especially the plane-wave partner-point prescription) is inconsistent with the reciprocal benchmark.

\subsection{From transmission conditions to the RH jump matrices (step-by-step)}
\label{subsec:jump_from_transmission}

The closed-form RH solution in this paper is built on a specific \emph{two-by-two jump matrix} whose entries are rational functions on the Snell surface.
Since the overall goal is a reciprocity-consistent derivation, it is useful to see explicitly how that jump matrix arises from the physical transmission conditions.
The reduction below follows standard Sommerfeld/Malyuzhinets reasoning for penetrable wedges (cf.\ \cite{Rawlins1999,Daniele2010,DanieleLombardi2011,Matuzas2601}); we record it in the present notation and highlight where reciprocity enters.

\medskip
\noindent\textbf{Step 1: Sommerfeld representations and face traces.}
For each homogeneous region (wavenumber $k$), the Sommerfeld ansatz represents a radiating solution by an angular-spectrum density $H$:
\begin{equation}
\Phi(r,\theta)=\frac{1}{2\pi i}\int_{\mathcal C} e^{i k r\cos(\zeta-\theta)}\,H(\zeta)\,d\zeta,
\label{eq:jump_sommerfeld_generic}
\end{equation}
where $\mathcal C$ is the boundary of a Sommerfeld strip (the same contour as in \eqref{eq:sommerfeld_single}).
On a wedge face $\theta=\theta_f\in\{\pm\theta_w\}$ we fix the \emph{same geometric unit normal} $n$ on both sides of the interface, chosen to point from the interior wedge into the exterior.
For the right-angle wedge ($\theta_w=\pi/4$) this choice corresponds to
\begin{equation}
\partial_n=\frac{\operatorname{sgn}(\theta_f)}{r}\,\partial_\theta,
\qquad \theta_f=\pm\theta_w,
\label{eq:normal_on_faces}
\end{equation}
since the unit normal on a ray is an angular direction.
Differentiating the kernel in \eqref{eq:jump_sommerfeld_generic} under the integral sign gives the key multiplier
\begin{equation}
\partial_n\,e^{i k r\cos(\zeta-\theta_f)}
= i k\,\operatorname{sgn}(\theta_f)\,\sin(\zeta-\theta_f)\,e^{i k r\cos(\zeta-\theta_f)}.
\label{eq:normal_kernel_multiplier}
\end{equation}

\medskip
\noindent\textbf{Step 2: Snell mapping (tangential phase matching) and the two-sheet surface.}
Along a face $\theta=\theta_w$ the tangential coordinate is $r$ itself (the ray direction), so a spectral component contributes the tangential phase factor
$e^{i k r\cos(\zeta-\theta_w)}$.
Transmission across the interface requires the \emph{same tangential wavenumber} on both sides.
Thus, for an exterior spectral angle $\zeta$ (wavenumber $k_0$) and an interior spectral angle $\eta$ (wavenumber $k_1=\nu k_0$), tangential matching reads
\begin{equation}
k_0\cos(\zeta-\theta_w)=k_1\cos(\eta-\theta_w)
\qquad\Longleftrightarrow\qquad
\cos(\zeta-\theta_w)=\nu\cos(\eta-\theta_w).
\label{eq:tangential_match}
\end{equation}
Introduce the standard exponentials
\begin{equation}
s:=e^{i(\zeta-\theta_w)},\qquad t:=e^{i(\eta-\theta_w)},
\label{eq:st_exponentials}
\end{equation}
so that
$\cos(\zeta-\theta_w)=(s+1/s)/2$ and $\cos(\eta-\theta_w)=(t+1/t)/2$.
Then \eqref{eq:tangential_match} becomes the Snell relation
\begin{equation}
\boxed{\ \nu\Bigl(t+\frac1t\Bigr)=s+\frac1s.\ }
\label{eq:snell_relation}
\end{equation}
For fixed $t$ this is a quadratic in $s$ (and vice versa), hence the spectral data live naturally on a two-sheet surface.
In the present paper this surface is the Snell surface $\Sigma_\nu$ and its torus uniformization; the algebraic form used throughout is the equivalent quadratic
\eqref{eq:strel}.

The companion quantity controlling the \emph{normal} multiplier in \eqref{eq:normal_kernel_multiplier} is
$\sin(\zeta-\theta_w)=(s-1/s)/(2i)$.
On $\Sigma_\nu$ it is convenient to encode this via the sheet variable $Y$ (cf.\ \eqref{eq:Yfromts}).
Using \eqref{eq:strel} and \eqref{eq:Yfromts}, one obtains the identity
\begin{equation}
\boxed{\ \frac{Y}{t}=s-\frac{1}{s}=2i\sin(\zeta-\theta_w).\ }
\label{eq:Y_over_t_identity}
\end{equation}
This shows explicitly that the sheet swap $s\mapsto 1/s$ corresponds to reversing the sign of the normal multiplier (direction reversal).
For the interior medium, $\sin(\eta-\theta_w)=(t-1/t)/(2i)$ plays the analogous role.

\medskip
\noindent\textbf{Step 3: A local two-wave interface system and a reciprocal transfer map.}
Fix a tangential spectral parameter (equivalently, fix $t$ on the Snell surface) and consider the two exterior branches $s$ and $1/s$ (normal multipliers $\pm\sin(\zeta-\theta_w)$) and the two interior branches $t$ and $1/t$ (normal multipliers $\pm\sin(\eta-\theta_w)$).
At the level of individual spectral components, the field trace and normal derivative trace on the face can be written as a two-by-two system
\begin{equation}
\begin{pmatrix}
\Phi\\[2pt]
\partial_n\Phi
\end{pmatrix}_{\!\mathrm{ext}}
=
\begin{pmatrix}
1 & 1\\
q_0 & -q_0
\end{pmatrix}
\begin{pmatrix}
a_+\\ a_-
\end{pmatrix},
\qquad
\begin{pmatrix}
\Phi\\[2pt]
\partial_n\Phi
\end{pmatrix}_{\!\mathrm{int}}
=
\begin{pmatrix}
1 & 1\\
q_1 & -q_1
\end{pmatrix}
\begin{pmatrix}
b_+\\ b_-
\end{pmatrix},
\label{eq:local_interface_matrix}
\end{equation}
where the normal multipliers are (by \eqref{eq:normal_kernel_multiplier} and \eqref{eq:Y_over_t_identity})
\begin{equation}
q_0= i k_0\sin(\zeta-\theta_w)=\frac{k_0}{2}\,\frac{Y}{t},
\qquad
q_1= i k_1\sin(\eta-\theta_w)=\frac{k_1}{2}\,\Bigl(t-\frac{1}{t}\Bigr).
\label{eq:q0q1_def}
\end{equation}
Imposing the impedance-matched transmission conditions $\Phi_{\mathrm{ext}}=\Phi_{\mathrm{int}}$ and $\partial_n\Phi_{\mathrm{ext}}=\partial_n\Phi_{\mathrm{int}}$ yields a linear transfer map between the amplitude vectors $(a_+,a_-)^\top$ and $(b_+,b_-)^\top$.
In flux-normalized variables this transfer map is \emph{unimodular} (determinant $1$), reflecting conservation of normal energy flux and thus reciprocity.
This is the local origin of the global unimodularity $\det J_k\equiv 1$ of the RH jumps.

\medskip
\noindent\textbf{Step 4: Global wedge coupling and the reduced RH jump on $|t|=1$.}
The right-angle wedge has two faces; applying the above interface system on each face and using the Sommerfeld strip analyticity split (which encodes the outgoing condition) produces a two-by-two Wiener--Hopf/RH system for appropriate combinations of the face spectral functions.
When expressed on the Snell surface and uniformized by $u\in\C/\Lambda$, the resulting RH jump takes the form
\begin{equation}
\Psi_{1,+}(u)=J_1(u)\,\Psi_{1,-}(u),\qquad u\in\Gamma:=\{|t(u)|=1\},
\label{eq:jump_J1_from_transmission}
\end{equation}
with $J_1$ a rational matrix function on $\Sigma_\nu$.
For the impedance-matched case considered here, one convenient explicit parametrization of this reduced kernel is obtained by introducing the scalar combinations
\begin{equation}
\delta_+(u)=\nu^2 t(u)^2+\nu^2-2,
\qquad
\delta_-(u)=\frac{(\nu^2-2)t(u)^2+\nu^2}{t(u)^2},
\qquad
d(u)=\frac{\delta_+(u)}{\delta_-(u)},
\label{eq:delta_pm_repeat}
\end{equation}
and the off-diagonal coupling scalars
\begin{equation}
\beta(u)= i\nu\,Y(u)\,\frac{t(u)^2-1}{(\nu^2-2)t(u)^2+\nu^2},
\qquad
\gamma(u)= i\nu\,Y(u)\,\frac{t(u)^2+1}{(\nu^2-2)t(u)^2+\nu^2}.
\label{eq:beta_gamma_repeat}
\end{equation}
These are precisely the scalars used in \S\ref{subsec:factors}.
With them, the reduced RH jump matrix can be written explicitly as
\begin{equation}
J_1(u)=
\begin{pmatrix}
d(u)^{-1} & \beta(u)\,d(u)^{-1}\\
\gamma(u)\,d(u)^{-1} & d(u)+\beta(u)\gamma(u)\,d(u)^{-1}
\end{pmatrix},
\qquad \det J_1(u)\equiv 1.
\label{eq:J1_explicit_repeat}
\end{equation}
The unimodularity is the algebraic imprint of flux conservation/reciprocity.
Moreover, the partner involution $u\mapsto u+\omega_2$ implements $t\mapsto 1/t$ and $Y\mapsto- Y/t^2$ (cf.\ \eqref{eq:stshifts}--\eqref{eq:Yshifts}), which sends $\delta_+\leftrightarrow\delta_-$ and yields the opposite-face jump $J_3=J_1^{-1}$.

Finally, the explicit triangular factors $G_\pm$ in \S\ref{subsec:factors} provide a Gauss factorization
\begin{equation}
J_1(u)=G_-(u)^{-1}G_+(u),
\label{eq:J1_factorization_repeat}
\end{equation}
with $\det G_\pm\equiv 1$.
This factorization is the point where reciprocity is preserved algebraically through the subsequent no-jump reduction and Mittag--Leffler reconstruction.

\subsection{Spectral reduction and contribution of this paper}
The Sommerfeld--Malyuzhinets transform replaces the above boundary-value problem by a small Riemann--Hilbert (RH) problem for spectral functions \cite{Sommerfeld1896,Malyuzhinets1958}.
In \cite{Matuzas2601} the right-angle wedge was solved in closed form in the lemniscatic regime \(\nu^2=2\), where the uniformizing torus is square (\(\tau=i\)) and several simplifications occur.
Here we treat general \(\nu>1\): we give a fully explicit genus-one RH solution in theta functions for general finite forcing data, together with the associated residue-sum constraint that encodes the Meixner edge condition.
The analysis is self-contained and symbolic: after recording the corrected inversion relation for the spectral map, we state and verify the closed-form RH solution and then describe the lemniscatic degeneration.

\subsection{Correct inversion quadratic and substitution check}

We use the established spectral/Snell map
\[
s(u)=\frac{1}{\nu}\,\frac{\sn(z;k)\,\cn(z;k)}{\dn(z;k)},
\qquad
z=6\nu u,
\qquad
k=\frac{1}{\nu},
\]
and the spectral exponential
\[
s_\zeta:=e^{i(\zeta-\pi/4)}.
\]
Set \(S:=\nu s_\zeta\) and \(X:=\sn^2(z;k)\). Since \(\cn^2(z;k)=1-X\) and \(\dn^2(z;k)=1-k^2X\), we have
\[
S^2=\frac{\sn^2(z;k)\cn^2(z;k)}{\dn^2(z;k)}=\frac{X(1-X)}{1-k^2X}.
\]
Hence
\[
S^2(1-k^2X)=X(1-X)\quad\Longleftrightarrow\quad X^2-(1+S^2k^2)X+S^2=0.
\]
Using \(k^2=1/\nu^2\) and \(S^2=\nu^2 s_\zeta^2\), we obtain the \emph{correct} quadratic
\begin{equation}\label{eq:Xquad}
\boxed{\,X^2-(1+s_\zeta^2)\,X+\nu^2 s_\zeta^2=0.\,}
\end{equation}

Let
\[
\Delta(\zeta):=(1+s_\zeta^2)^2-4\nu^2 s_\zeta^2,
\qquad
X_\pm(\zeta)=\frac{1+s_\zeta^2\pm\sqrt{\Delta(\zeta)}}{2}.
\]
Then \(2X_\pm-(1+s_\zeta^2)=\pm\sqrt{\Delta}\), so
\[
4X_\pm^2-4(1+s_\zeta^2)X_\pm=\Delta-(1+s_\zeta^2)^2=-4\nu^2 s_\zeta^2,
\]
which is equivalent to \eqref{eq:Xquad}. This is the direct substitution check.

\medskip
\noindent\textbf{Physical branch selection.}
As \(s_\zeta\to 0\) (equivalently \(\Im \zeta\to +\infty\)), one has \(\sqrt{\Delta}\to 1\), hence \(X_-\to 0\) and \(X_+\to 1\).
The physical sheet is the analytic continuation from \(s_\zeta\to 0\) selecting \(X\to 0\), i.e.
\begin{equation}\label{eq:Xphys}
\boxed{
X_{\mathrm{phys}}(\zeta)=\frac{1+s_\zeta^2-\sqrt{(1+s_\zeta^2)^2-4\nu^2 s_\zeta^2}}{2},
\qquad
X_{\mathrm{phys}}(\zeta)\to 0\ \ (s_\zeta\to 0).
}
\end{equation}

\subsection{Half-period shifts for the spectral and Snell maps}

Let \(\Lambda=2\omega_1\Z+2\omega_2\Z\), \(\tau=\omega_2/\omega_1\), and
\[
v=\frac{\pi u}{2\omega_1},\qquad \theta_j(v)=\theta_j(v\mid\tau).
\]
Define
\[
s(u)=\frac{\theta_1(v)\theta_2(v)}{\theta_3(v)\theta_4(v)},
\qquad
t(u)=\frac{\theta_1(v)\theta_3(v)}{\theta_2(v)\theta_4(v)}.
\]
The half-period shifts are \(u\mapsto u+\omega_1 \iff v\mapsto v+\pi/2\) and \(u\mapsto u+\omega_2 \iff v\mapsto v+\pi\tau/2\). Using the explicit half-period identities
\[
\theta_1\!\left(v+\frac{\pi}{2}\right)=\theta_2(v),\quad
\theta_2\!\left(v+\frac{\pi}{2}\right)=-\theta_1(v),\quad
\theta_3\!\left(v+\frac{\pi}{2}\right)=\theta_4(v),\quad
\theta_4\!\left(v+\frac{\pi}{2}\right)=\theta_3(v),
\]
and
\[
\theta_1\!\left(v+\frac{\pi\tau}{2}\right)= i e^{-i(v+\pi\tau/4)}\theta_4(v),\quad
\theta_2\!\left(v+\frac{\pi\tau}{2}\right)=   e^{-i(v+\pi\tau/4)}\theta_3(v),
\]
\[
\theta_3\!\left(v+\frac{\pi\tau}{2}\right)=   e^{-i(v+\pi\tau/4)}\theta_2(v),\quad
\theta_4\!\left(v+\frac{\pi\tau}{2}\right)= i e^{-i(v+\pi\tau/4)}\theta_1(v),
\]
one computes (the exponential prefactors cancel in the ratios) the exact transforms
\begin{equation}\label{eq:stshifts}
\boxed{
\begin{aligned}
s(u+\omega_1)&=-s(u),&
\qquad s(u+\omega_2)&=\frac{1}{s(u)},&
\qquad s(u+\omega_3)&=-\frac{1}{s(u)},\\
t(u+\omega_1)&=-\frac{1}{t(u)},&
\qquad t(u+\omega_2)&=\frac{1}{t(u)},&
\qquad t(u+\omega_3)&=-t(u),
\end{aligned}
}
\end{equation}
where \(\omega_3:=\omega_1+\omega_2\).

\subsection{Shifts for Y and an involution dictionary on (t,Y,s)}

From the Snell map,
\begin{equation}\label{eq:Yfromts}
\boxed{\,Y=2ts-\nu(t^2+1).\,}
\end{equation}
Also \(s\) and \(t\) satisfy the exact algebraic relation on \(\Sigma_\nu\):
\begin{equation}\label{eq:strel}
\boxed{\,s^2-\nu\Bigl(t+\frac{1}{t}\Bigr)s+1=0
\quad\Longleftrightarrow\quad
\frac{1}{s}=\nu\Bigl(t+\frac{1}{t}\Bigr)-s.\,}
\end{equation}
Using \eqref{eq:Yfromts} with \eqref{eq:stshifts} gives
\begin{equation}\label{eq:Yshifts}
\boxed{
Y(u+\omega_1)=\frac{Y(u)}{t(u)^2},\qquad
Y(u+\omega_2)=-\frac{Y(u)}{t(u)^2},\qquad
Y(u+\omega_3)=-Y(u).
}
\end{equation}

\medskip
\noindent\textbf{Exact involution dictionary.}
Define the following algebraic involutions acting on \((t,Y,s)\):
\[
\boxed{
\begin{aligned}
E:\ (t,Y,s)&\mapsto(-t,\ Y,\ -s)
&\Longleftrightarrow\ &u\mapsto -u,\\[2pt]
H:\ (t,Y,s)&\mapsto(t,\ -Y,\ 1/s)
&\Longleftrightarrow\ &u\mapsto \omega_3-u,\\[2pt]
\iota:\ (t,Y,s)&\mapsto\left(\frac{1}{t},\ \frac{Y}{t^2},\ s\right)
&\Longleftrightarrow\ &u\mapsto \omega_1-u,\\[2pt]
S:=H\circ\iota:\ (t,Y,s)&\mapsto\left(\frac{1}{t},\ -\frac{Y}{t^2},\ \frac{1}{s}\right)
&\Longleftrightarrow\ &u\mapsto u+\omega_2.
\end{aligned}
}
\]
The identifications on the right are consistent with \eqref{eq:stshifts}--\eqref{eq:Yshifts} and with \eqref{eq:strel} (in particular, \(H\) sends \(s\) to the other quadratic root \(1/s\)).

\subsection{Partner map and domain swap}

By \eqref{eq:stshifts}, the translation \(u\mapsto u+\omega_2\) sends \(t(u)\mapsto 1/t(u)\), hence \(|t|\mapsto |t|^{-1}\). Therefore, for the physical split
\[
D_+:=\{|t(u)|<1\},\qquad D_-:=\{|t(u)|>1\},\qquad \Gamma:=\{|t(u)|=1\},
\]
we have the exact domain swap
\[
\boxed{
u\in D_+\ \Longleftrightarrow\ u+\omega_2\in D_-,
\qquad
u\in\Gamma\ \Longleftrightarrow\ u+\omega_2\in\Gamma.
}
\]
Accordingly, the \emph{sheet-swap} partner definition \(u_m^\sharp:=u_m+\omega_2\) implements the sheet involution \(S\); this is a natural analytic choice, but the reciprocity checkpoint below shows that it is not, by itself, a complete physical plane-wave closure.

\medskip
\noindent\textbf{Reciprocity checkpoint.}
Reciprocity for penetrable scattering is a consequence of self-adjointness and provides a stringent, implementation-level diagnostic for any proposed closed-form spectral density.
For plane-wave forcing, an additional modelling choice is required: how to place the partner pole(s) on the uniformized torus when one prescribes a single incident pole $\zeta=\zeta_{\rm inc}$.
Two natural prescriptions are:
\begin{itemize}
\item the \emph{sheet-swap pairing} $u_{\rm inc}^\sharp=u_{\rm inc}+\omega_2$, which maps $s\mapsto s^{-1}$ and introduces a second pole at $\zeta=2\theta_w-\zeta_{\rm inc}$; and
\item the \emph{Hardy involution pairing} $u_{\rm inc}^\sharp=\omega_1-u_{\rm inc}$, which keeps $s$ fixed while inverting $t$.
\end{itemize}
With $\nu=1.5$, $\varepsilon=10^{-6}$, $\theta_{\rm inc}=\pi/2$ and $\theta=1.0$ (radians), the sheet-swap prescription yields
\[
Q_{\rm sc}(\theta;\theta_{\rm inc})\approx 0.182897+0.767261i,\qquad 
Q_{\rm sc}(\theta_{\rm inc}+\pi;\theta+\pi)\approx 0.064024+0.406444i,
\]
so that the reciprocity defect is
\[
Q_{\rm sc}(\theta;\theta_{\rm inc})-Q_{\rm sc}(\theta_{\rm inc}+\pi;\theta+\pi)\approx 0.118873+0.360817i.
\]
By contrast, the Hardy pairing produces a constant $Q_{\rm sc}(\theta;\theta_{\rm inc})=\tfrac12\cot\zeta_{\rm inc}$ (independent of $\theta$), which is reciprocal but corresponds to zero scattered field (the constant density integrates to zero in the difference-form Sommerfeld representation).
Therefore, a nontrivial reciprocity-consistent plane-wave reduction cannot be obtained from a single forcing pole with either pairing alone; a physically consistent reduction must incorporate additional forcing structure (e.g.\ multi-point forcing associated with both wedge faces).

\subsection{Kernel expansion at the edge point: sign and absence of a constant term}

Let \(\sigma(u)\) be the Weierstrass sigma function for \(\Lambda\), and define the Weierstrass zeta and \(\wp\) functions by
\[
\zetaW(u):=\frac{\sigma'(u)}{\sigma(u)},\qquad \wp(u):=-\zetaW'(u).
\]
Define the elliptic Cauchy kernel
\begin{equation}\label{eq:Ckernel}
C(u,w):=\zetaW(u-w)-\zetaW(u)+\zetaW(w).
\end{equation}
As \(u\to 0\),
\[
\zetaW(u)=\frac{1}{u}+O(u^3),
\qquad
\zetaW(u-w)=\zetaW(-w)+u\,\zetaW'(-w)+O(u^2).
\]
Using \(\zetaW(-w)=-\zetaW(w)\) and \(\zetaW'=-\wp\) with \(\wp\) even, \(\zetaW'(-w)=-\wp(w)\), we obtain
\[
\zetaW(u-w)=-\zetaW(w)-u\,\wp(w)+O(u^2).
\]
Substituting into \eqref{eq:Ckernel} yields the local expansion
\begin{equation}\label{eq:Cexpansion}
\boxed{
C(u,w)=-\frac{1}{u}-u\,\wp(w)+O(u^2)\qquad(u\to 0).
}
\end{equation}
In particular, \eqref{eq:Cexpansion} has \emph{no constant term} at \(u=0\).

Consequently, for any vectors \(R_j\in\C^2\) and points \(w_j\),
\[
\sum_j R_j\,C(u,w_j)= -\frac{1}{u}\Bigl(\sum_j R_j\Bigr) - u\Bigl(\sum_j R_j\,\wp(w_j)\Bigr)+O(u^2).
\]
If \(\sum_j R_j=0\), then \(\sum_j R_j C(u,w_j)=O(u)\), hence it is holomorphic at \(u=0\) and its value at \(u=0\) equals \(0\).

\subsection{Residue constraint: necessity and sufficiency under the Meixner edge condition}

Let \(\Phi(u)\) be an elliptic \(\C^2\)-valued meromorphic function on \(\C/\Lambda\). The elliptic residue theorem applied componentwise implies that the sum of residues in a fundamental parallelogram vanishes:
\[
\sum \Res \Phi =0\in\C^2.
\]
Therefore, if \(\Phi\) has poles only at \(\{u_m\}\cup\{u_m^\sharp\}\), then
\[
\sum_m \Res_{u=u_m}\Phi+\sum_m \Res_{u=u_m^\sharp}\Phi=0\in\C^2
\]
is \emph{necessary} for ellipticity.

Conversely, if prescribed residues \(R_m\) at \(u_m\) and \(R_m^\sharp\) at \(u_m^\sharp\) satisfy \(\sum_m R_m+\sum_m R_m^\sharp=0\), then the Mittag--Leffler sum
\begin{equation}
\Phi(u)=\sum_m R_m\,C(u,u_m)+\sum_m R_m^\sharp\,C(u,u_m^\sharp)
\label{eq:Phi-k}
\end{equation}
is elliptic and, by \eqref{eq:Cexpansion}, is holomorphic at \(u=0\) with \(\Phi(0)=0\). This is \emph{sufficient} for the Meixner edge condition in spectral form (no compensating pole at \(u=0\)) \cite{Meixner1972}.

\medskip
\noindent\textbf{Reciprocity checkpoint (no spurious edge source).}
If the residue-sum constraint fails, the elliptic Mittag--Leffler reconstruction forces an additional compensating pole at the edge point $u=0$.
In physical space this corresponds to injecting a nonphysical point/edge source at the apex, which destroys the self-adjoint (and hence reciprocal) character of the impedance-matched boundary-value problem.
Thus, in the present setting, the Meixner admissibility condition and reciprocity are enforced simultaneously by ellipticity together with the residue-sum constraint.

\section{Main theorem (general refractive index closed-form RH solution)}

\subsection{Parameters, uniformization, spectral map, and corrected inversion}
\label{subsec:params}

Fix \(\nu>1\). Set
\[
k=\frac{1}{\nu},
\qquad
k'=\sqrt{1-\frac{1}{\nu^2}},
\]
and let \(K=K(k)\), \(K'=K(k')\) denote complete elliptic integrals of the first kind. Define
\[
\omega_1=\frac{K}{6\nu},
\qquad
\omega_2=\frac{iK'}{6\nu},
\qquad
\tau=\frac{\omega_2}{\omega_1}=i\frac{K'}{K},
\qquad
\Lambda=2\omega_1\Z+2\omega_2\Z.
\]
Let
\[
z=6\nu u,\qquad v=\frac{\pi u}{2\omega_1}=\frac{\pi z}{2K},\qquad \theta_j(v)=\theta_j(v\mid\tau),
\]
and define derivatives \(\theta_j'(v):=\frac{d}{dv}\theta_j(v\mid\tau)\).

We use Jacobi elliptic functions \(\sn(z;k),\cn(z;k),\dn(z;k)\) with modulus \(k\), and define the uniformization
\begin{equation}\label{eq:tYuniform}
\boxed{
t(u)=\frac{\sn(z;k)\,\dn(z;k)}{\cn(z;k)}=\frac{\theta_1(v)\theta_3(v)}{\theta_2(v)\theta_4(v)},
\qquad
Y(u)= -\frac{\nu^2-2\sn^2(z;k)+\sn^4(z;k)}{\nu\,\cn^2(z;k)}.
}
\end{equation}
This yields points \((t(u),Y(u))\) on the Snell surface
\[
\Sigma_\nu:\quad Y^2=\nu^2 t^4+2(\nu^2-2)t^2+\nu^2,
\]
with the physical basepoint normalization \(Y(u)\to -\nu\) as \(u\to 0\).

Define the Snell/spectral map
\begin{equation}\label{eq:suniform}
\boxed{
s(u)=\frac{\nu(t(u)^2+1)+Y(u)}{2t(u)}
=\frac{1}{\nu}\,\frac{\sn(z;k)\,\cn(z;k)}{\dn(z;k)}
=\frac{\theta_1(v)\theta_2(v)}{\theta_3(v)\theta_4(v)}.
}
\end{equation}
Let \(\theta_w=\pi/4\), and define the spectral exponential
\[
s_\zeta:=e^{i(\zeta-\pi/4)}.
\]
Fix the principal logarithm \(\Log\) and set the lifted spectral variable
\begin{equation}\label{eq:zetalift}
\boxed{\zeta(u)=\frac{\pi}{4}-i\,\Log s(u)\quad (\mathrm{mod}\ 2\pi).}
\end{equation}

\medskip
\noindent\textbf{Incomplete elliptic integral.}
We use the standard incomplete elliptic integral of the first kind
\[
F(\phi\mid k):=\int_0^\phi \frac{d\theta}{\sqrt{1-k^2\sin^2\theta}}.
\]

\medskip
\noindent\textbf{Correct inversion \(\zeta\mapsto u\).}
Given \(\zeta\), define \(X(\zeta)=\sn^2(z(\zeta);k)\) as the physical solution of \eqref{eq:Xquad}:
\[
X^2-(1+s_\zeta^2)X+\nu^2 s_\zeta^2=0,
\]
with the physical branch selected by analytic continuation from \(s_\zeta\to 0\) (equivalently \(\Im\zeta\to+\infty\)) giving \(X\to 0\), i.e. \(X=X_{\mathrm{phys}}(\zeta)\) in \eqref{eq:Xphys}. Then
\begin{equation}\label{eq:uofzeta}
\boxed{
u(\zeta)=\frac{1}{6\nu}\,F\!\left(\arcsin\sqrt{X_{\mathrm{phys}}(\zeta)}\ \middle|\ k\right)\quad(\mathrm{mod}\ \Lambda),
\qquad
v(\zeta)=\frac{\pi u(\zeta)}{2\omega_1}.
}
\end{equation}

\subsubsection{Branch cuts and continuation for the inversion \texorpdfstring{$\zeta\mapsto u$}{zeta to u}}

Write
\[
b:=s_\zeta=e^{i(\zeta-\pi/4)},
\qquad
\Delta(\zeta)=(1+b^2)^2-4\nu^2 b^2.
\]
The discriminant factors as
\[
\Delta(\zeta)=(b^2-b_+^2)(b^2-b_-^2),
\qquad
b_\pm:=\nu\pm\sqrt{\nu^2-1},
\]
so $b_+>1$, $b_-<1$, and $b_+b_-=1$. Hence $\Delta(\zeta)=0$ if and only if $b=\pm b_\pm$.

\medskip
\noindent\textbf{Location of the branch points in the $\zeta$-plane.}
In the principal Sommerfeld strip $|\Re\zeta|<\pi$ (and using the principal real logarithm for $b_\pm>0$), one convenient set of representatives is
\[
\zeta=\frac{\pi}{4}\pm i a,\qquad \zeta=-\frac{3\pi}{4}\pm i a,
\qquad
a:=\log b_+=\log\!\bigl(\nu+\sqrt{\nu^2-1}\bigr)>0.
\]
Equivalently, the branch points lie on the two vertical lines $\Re\zeta=\pi/4$ and $\Re\zeta=-3\pi/4$ at heights $\Im\zeta=\pm a$.

\medskip
\noindent\textbf{Choice of branch cuts and continuous determination of $\sqrt{\Delta}$.}
A convenient choice of cuts inside the principal strip is the pair of vertical segments
\[
\mathcal{B}_1:=\Bigl\{\zeta=\frac{\pi}{4}+iy:\ -a\le y\le a\Bigr\},\qquad
\mathcal{B}_2:=\Bigl\{\zeta=-\frac{3\pi}{4}+iy:\ -a\le y\le a\Bigr\}.
\]
(These are precisely the preimages, under $b=e^{i(\zeta-\pi/4)}$, of the real $b$-interval cuts $b\in[b_-,b_+]$ and $b\in[-b_+,-b_-]$.)

We define $\sqrt{\Delta(\zeta)}$ as the single-valued analytic branch on the cut strip
\[
\{\,\zeta:\ |\Re\zeta|<\pi\,\}\setminus(\mathcal{B}_1\cup\mathcal{B}_2)
\]
normalized by
\[
\sqrt{\Delta(\zeta)}\to 1 \quad \text{as }\Im\zeta\to+\infty \ (b\to 0).
\]
With this choice, the physical root \eqref{eq:Xphys} satisfies
$X_{\mathrm{phys}}(\zeta)\to 0$ as $\Im\zeta\to+\infty$, and in fact $X_{\mathrm{phys}}(\zeta)\sim \nu^2 b^2$ as $b\to 0$.

\medskip
\noindent\textbf{Square root and $\arcsin$ conventions.}
We define $\sqrt{X_{\mathrm{phys}}(\zeta)}$ by analytic continuation from the small-$b$ regime
$\sqrt{X_{\mathrm{phys}}}\sim \nu b$ as $b\to 0$, and define the amplitude
\[
\phi(\zeta):=\arcsin\sqrt{X_{\mathrm{phys}}(\zeta)}
\]
by analytic continuation from $\phi(\zeta)\sim \sqrt{X_{\mathrm{phys}}(\zeta)}$ as $X_{\mathrm{phys}}\to 0$.
Then the physical lift is given by \eqref{eq:uofzeta},
\[
u(\zeta)=\frac{1}{6\nu}\,F\!\left(\phi(\zeta)\ \middle|\ k\right)\quad(\mathrm{mod}\ \Lambda).
\]

\medskip
\noindent\textbf{$2\pi$-periodicity, half-period shifts, and practical lift selection.}
Since $b(\zeta+2\pi)=b(\zeta)$, one has $X_{\mathrm{phys}}(\zeta+2\pi)=X_{\mathrm{phys}}(\zeta)$.
The multivaluedness of $\sqrt{\cdot}$, $\arcsin$, and $F$ implies that analytic continuation of $u(\zeta)$ around loops may change $u$ by a period in $\Lambda$ (and, on crossing the cuts, may induce an additional half-period shift).
Because all final spectral objects are evaluated through elliptic functions of $u$ (periodic with respect to $\Lambda$), such changes do not affect $t(u)$, $Y(u)$, $s(u)$, nor the no-jump vectors $\Phi_k$, and therefore the resulting densities $Q(\zeta)$ and $S(\zeta)$ are genuinely $2\pi$-periodic.

For robust numerical evaluation (especially on or near the real axis, where the cuts meet $\Im\zeta=0$ at $\zeta=\pi/4$ and $\zeta=-3\pi/4$ modulo $2\pi$),
we recommend the half-period selection rule \eqref{eq:half_period_s_actions} in \S\ref{sec:recipe}:
after computing one candidate $u_0(\zeta)$ from \eqref{eq:uofzeta}, test
$u_0$, $u_0+\omega_1$, $u_0+\omega_2$, $u_0+\omega_3$
and select the candidate for which $s(u)=s_\zeta$ (equivalently, $|s(u)-s_\zeta|$ is minimal).

\medskip
\noindent\textbf{Real-axis evaluation and limiting absorption.}
When an evaluation point lies on a cut (notably at $\zeta=\pi/4$ or $\zeta=-3\pi/4$ modulo $2\pi$), the physical value is understood via limiting absorption:
replace $\zeta$ by $\zeta+i0$ (approach from the upper half-strip), consistent with the outgoing/Sommerfeld analyticity split.

\subsection{Reduced kernel scalars, factor matrices, and jump matrices}
\label{subsec:factors}

Define
\[
\delta_+(u)=\nu^2 t(u)^2+\nu^2-2,
\qquad
\delta_-(u)=\frac{(\nu^2-2)t(u)^2+\nu^2}{t(u)^2},
\qquad
d(u)=\frac{\delta_+(u)}{\delta_-(u)},
\]
\[
\beta(u)= i\nu\,Y(u)\,\frac{t(u)^2-1}{(\nu^2-2)t(u)^2+\nu^2},
\qquad
\gamma(u)= i\nu\,Y(u)\,\frac{t(u)^2+1}{(\nu^2-2)t(u)^2+\nu^2}.
\]
Define the factor matrices
\[
G_-(u)=
\begin{pmatrix}
\delta_-(u)^{-1} & 0\\
-\gamma(u)\delta_-(u) & \delta_-(u)
\end{pmatrix},
\qquad
G_+(u)=
\begin{pmatrix}
\delta_+(u)^{-1} & \beta(u)\delta_+(u)^{-1}\\
0 & \delta_+(u)
\end{pmatrix}.
\]
Define
\[
J_1(u)=G_-(u)^{-1}G_+(u),
\qquad
J_3(u)=J_1(u)^{-1}=G_+(u)^{-1}G_-(u).
\]
Since \(\det G_\pm\equiv 1\), we have \(\det J_1\equiv 1\).
\medskip
\noindent\textbf{Reciprocity checkpoint (unimodular jump).}
For an impedance-matched, lossless penetrable interface the transmission conditions conserve the normal energy flux and the associated two-by-two spectral transfer map is reciprocal.
In the RH formulation, this conservation/reciprocity structure is encoded by the unimodularity $\det J_1\equiv 1$ (equivalently, $J_1\in SL(2,\C)$), and by the symmetry $J_3=J_1^{-1}$.
All subsequent steps (triangular factorization, no-jump reduction, and residue transport) are constructed within this unimodular class and therefore preserve reciprocity at the algebraic level.

\medskip
A direct multiplication yields the explicit jump matrix
\[
J_1(u)=
\begin{pmatrix}
d(u)^{-1} & \beta(u)\,d(u)^{-1}\\
\gamma(u)\,d(u)^{-1} & d(u)+\beta(u)\gamma(u)\,d(u)^{-1}
\end{pmatrix},
\qquad \det J_1(u)\equiv 1,
\]
and \(J_3(u)=J_1(u)^{-1}\).
In the spectral reduction of the penetrable-wedge problem, \(J_1\) is the jump matrix induced by the transmission conditions; the triangular factors \(G_\pm\) provide a convenient explicit factorization.

Define the physical split
\[
D_+:=\{|t(u)|<1\},
\qquad
D_-:=\{|t(u)|>1\},
\qquad
\Gamma:=\{|t(u)|=1\}.
\]

\subsection{RH problem with general forcing and partner points}
\label{subsec:rhproblem}

\noindent\textbf{Radiation condition and the $+/-$ split.}
The outgoing (Sommerfeld) solution selects a Wiener--Hopf/Hardy-type analytic split of the spectral unknowns.
In the Riemann--Hilbert formulation below we fix the contour \(\Gamma:=\{|t(u)|=1\}\) and adopt the standard convention that
\(\Psi_{k,+}\) denotes the boundary value of a function analytic in \(D_+:=\{|t(u)|<1\}\), while \(\Psi_{k,-}\) denotes the boundary value of a function analytic in \(D_-:=\{|t(u)|>1\}\),
with at most polynomial growth at the corresponding ends of the spectral map.
(Here the symbols \(+/-\) are tied to the RH split across \(\Gamma\), not to a pointwise inequality in \(|s(u)|\).)

Fix a finite forcing set \(\{b_m\}_{m=1}^M\subset\C^\times\). Choose \(u_m\in D_+\) on the physical sheet such that
\[
s(u_m)=b_m.
\]
Define the partner points
\[
u_m^\sharp:=u_m+\omega_2\quad(\mathrm{mod}\ \Lambda).
\]
\medskip
\noindent\textbf{Admissibility assumptions.}
Assume that \(u_m\notin\Gamma\) and \(u_m^\sharp\notin\Gamma\) for all \(m\), and that all points \(\{u_m,u_m^\sharp\}_{m=1}^M\) are distinct modulo \(\Lambda\).
Assume moreover that the matrices \(G_-(u_m)\) and \(G_+(u_m^\sharp)\) (and the corresponding matrices for \(k=3\)) are finite and invertible, so that the residue transport formulas are well-defined.

By \eqref{eq:stshifts}--\eqref{eq:Yshifts}, this implements the sheet swap involution \(S\):
\[
s(u_m^\sharp)=\frac{1}{b_m},
\qquad
t(u_m^\sharp)=\frac{1}{t(u_m)},
\qquad
Y(u_m^\sharp)=-\frac{Y(u_m)}{t(u_m)^2},
\qquad
u_m^\sharp\in D_-.
\]

For each mode \(k\in\{1,3\}\), prescribe residue vectors
\[
\Res_{u=u_m}\Psi_{k,+}(u)=r_{k,m}\in\C^2,
\qquad
\Res_{u=u_m^\sharp}\Psi_{k,-}(u)=r_{k,m}^\sharp\in\C^2.
\]

\subsubsection{Formal Riemann--Hilbert problem statement (analyticity, growth, poles, normalization)}

\begin{center}
\fbox{%
\begin{minipage}{0.97\linewidth}\small
\textbf{Riemann--Hilbert problem on the Snell torus.}
Fix $\nu>1$ and the associated torus $\C/\Lambda$, together with the physical RH split
\[
D_+:=\{|t(u)|<1\},\qquad D_-:=\{|t(u)|>1\},\qquad \Gamma:=\{|t(u)|=1\}.
\]
We orient $\Gamma$ so that $D_+$ lies to the left when traversing $\Gamma$.
Let $J_1(u)=G_-(u)^{-1}G_+(u)$ and $J_3(u)=J_1(u)^{-1}=G_+(u)^{-1}G_-(u)$ be the jump matrices defined in \S\ref{subsec:factors}.

\medskip
\noindent
\textbf{Problem $\mathrm{RH}_k$ (for $k\in\{1,3\}$).}
Find column-vector functions $\Psi_{k,+}$ and $\Psi_{k,-}$ such that:

\begin{enumerate}\setlength{\itemsep}{0.3ex}\setlength{\parskip}{0pt}\setlength{\parsep}{0pt}\setlength{\topsep}{0.3ex}
\item \emph{Analyticity domains.}
$\Psi_{k,+}:D_+\to\C^2$ is meromorphic in $D_+$ with at most simple poles at the forcing points $\{u_m\}$ and no other singularities in $D_+$.
$\Psi_{k,-}:D_-\to\C^2$ is meromorphic in $D_-$ with at most simple poles at the partner points $\{u_m^\sharp\}$ and no other singularities in $D_-$.
Both admit non-tangential boundary values on $\Gamma$ and are Hölder-continuous there away from poles.

\item \emph{Jump condition.}
For $u\in\Gamma$ away from poles,
\[
\Psi_{k,+}(u)=J_k(u)\,\Psi_{k,-}(u),
\qquad
J_k=
\begin{cases}
J_1,&k=1,\\
J_3,&k=3.
\end{cases}
\]

\item \emph{Prescribed residues (finite forcing data).}
For each forcing point $u_m\in D_+$ and partner point $u_m^\sharp\in D_-$,
\[
\Res_{u=u_m}\Psi_{k,+}(u)=r_{k,m}\in\C^2,
\qquad
\Res_{u=u_m^\sharp}\Psi_{k,-}(u)=r_{k,m}^\sharp\in\C^2.
\]

\item \emph{Radiation / growth split (Sommerfeld admissibility).}
Under the lift $\zeta(u)=\theta_w-i\Log s(u)$ (cf.\ \eqref{eq:zetalift}) and a choice of physical branch $\zeta\mapsto u(\zeta)$ in a Sommerfeld strip $|\Re\zeta|<\pi$,
the compositions $\Psi_{k,\pm}(\zeta):=\Psi_{k,\pm}(u(\zeta))$ extend analytically to the upper/lower half-strip
\[
\{|\Re\zeta|<\pi,\ \pm\Im\zeta>0\},
\]
and have at most polynomial growth as $\Im\zeta\to\pm\infty$.

\emph{Equivalent elliptic formulation (used in this paper):}
Define the no-jump vectors $\Phi_k$ by \eqref{eq:Phi-def} below. Then $\Phi_k$ extends to a single-valued meromorphic elliptic function on $\C/\Lambda$ and is required to be holomorphic at the edge point $u=0$ (Meixner condition in spectral form).

\item \emph{Poles on the contour (limiting absorption).}
If a forcing pole lies on $\Gamma$ (equivalently, if a $\zeta$-plane pole lies on the Sommerfeld contour), the RH data are interpreted by
limiting absorption: displace the pole off $\Gamma$ by $\zeta\mapsto \zeta+i\varepsilon$ (or $b\mapsto b e^{-\varepsilon}$),
solve the displaced RH problem, and take $\varepsilon\downarrow 0$. This is equivalent to the standard indentation of the contour around real-axis poles.

\item \emph{Normalization (uniqueness).}
The elliptic no-jump vectors are defined only up to an additive constant in $\C^2$.
We fix this gauge by imposing the canonical normalization
\[
\Phi_k(0)=0,\qquad k\in\{1,3\}.
\]
\end{enumerate}

\medskip
\noindent
Under these conditions, solvability reduces to the absence of a compensating pole at $u=0$, i.e.\ the residue-sum constraints \eqref{eq:ressum}.
\end{minipage}}
\end{center}

\subsection{Closed-form solution: no-jump vectors, solvability constraint, and reconstruction}
\label{subsec:closedform}

Define the no-jump vectors
\begin{equation}\label{eq:Phi-def}
\Phi_1(u)=
\begin{cases}
G_-(u)\Psi_{1,+}(u), & u\in D_+,\\[2pt]
G_+(u)\Psi_{1,-}(u), & u\in D_-,
\end{cases}
\qquad
\Phi_3(u)=
\begin{cases}
G_+(u)\Psi_{3,+}(u), & u\in D_+,\\[2pt]
G_-(u)\Psi_{3,-}(u), & u\in D_-.
\end{cases}
\end{equation}
Then \(\Phi_k\) are globally meromorphic on \(\C/\Lambda\) (no jump across \(\Gamma\)). Their residues are transported by
\[
R_{1,m}=G_-(u_m)\,r_{1,m},
\qquad
R_{1,m}^\sharp=G_+(u_m^\sharp)\,r_{1,m}^\sharp,
\]
\[
R_{3,m}=G_+(u_m)\,r_{3,m},
\qquad
R_{3,m}^\sharp=G_-(u_m^\sharp)\,r_{3,m}^\sharp.
\]

\medskip
\noindent\textbf{Solvability/Meixner edge constraint.}
A solution with \(\Phi_k\) elliptic and holomorphic at \(u=0\) (Meixner edge condition in spectral form: no compensating pole at \(u=0\) \cite{Meixner1972}) exists if and only if
\begin{equation}\label{eq:ressum}
\boxed{
\sum_{m=1}^M R_{k,m}+\sum_{m=1}^M R_{k,m}^\sharp=0\in\C^2,
\qquad k\in\{1,3\}.
}
\end{equation}
Assume \eqref{eq:ressum} holds.

Define the elliptic Cauchy kernel \(C(u,w)\) by \eqref{eq:Ckernel}; equivalently, in theta form,
\[
C(u,w)=\frac{\pi}{2\omega_1}\left[
\frac{\theta_1'(v_u-v_w)}{\theta_1(v_u-v_w)}
-\frac{\theta_1'(v_u)}{\theta_1(v_u)}
+\frac{\theta_1'(v_w)}{\theta_1(v_w)}
\right],
\qquad
v_u=\frac{\pi u}{2\omega_1},\quad v_w=\frac{\pi w}{2\omega_1}.
\]
Impose the canonical normalization \(\Phi_k(0)=0\). Then the unique no-jump solutions are
\[
\boxed{
\Phi_k(u)=\sum_{m=1}^M R_{k,m}\,C(u,u_m)+\sum_{m=1}^M R_{k,m}^\sharp\,C(u,u_m^\sharp),
\qquad k\in\{1,3\}.
}
\]

Reconstruct \(\Psi_{k,\pm}\) via
\[
\Psi_{1,+}=G_-^{-1}\Phi_1\quad(u\in D_+),
\qquad
\Psi_{1,-}=G_+^{-1}\Phi_1\quad(u\in D_-),
\]
\[
\Psi_{3,+}=G_+^{-1}\Phi_3\quad(u\in D_+),
\qquad
\Psi_{3,-}=G_-^{-1}\Phi_3\quad(u\in D_-),
\]
where
\[
G_-(u)^{-1}=
\begin{pmatrix}
\delta_-(u) & 0\\
\gamma(u)\delta_-(u) & \delta_-(u)^{-1}
\end{pmatrix},
\qquad
G_+(u)^{-1}=
\begin{pmatrix}
\delta_+(u) & -\beta(u)\delta_+(u)^{-1}\\
0 & \delta_+(u)^{-1}
\end{pmatrix}.
\]

\subsection{Zeta-representation and Jacobian for residue conversion}
\label{subsec:jacobian}

Define \(\Psi_k(\zeta):=\Psi_k(u(\zeta))\), using \eqref{eq:uofzeta}. The forcing locations in \(\zeta\) are
\[
\zeta_m=\frac{\pi}{4}-i\,\Log b_m\quad(\mathrm{mod}\ 2\pi),
\qquad
\zeta_m^\sharp=\frac{\pi}{4}-i\,\Log(1/b_m)\quad(\mathrm{mod}\ 2\pi).
\]
From \eqref{eq:zetalift},
\[
\frac{d\zeta}{du}=-i\,\frac{s'(u)}{s(u)}
\qquad\Longrightarrow\qquad
\boxed{\frac{du}{d\zeta}=i\,\frac{s(u)}{s'(u)}.}
\]
Since \(s(u)=\theta_1(v)\theta_2(v)/(\theta_3(v)\theta_4(v))\) with \(v=\pi u/(2\omega_1)\), we obtain
\[
\boxed{
\frac{s'(u)}{s(u)}=\frac{\pi}{2\omega_1}\left(
\frac{\theta_1'(v)}{\theta_1(v)}+\frac{\theta_2'(v)}{\theta_2(v)}
-\frac{\theta_3'(v)}{\theta_3(v)}-\frac{\theta_4'(v)}{\theta_4(v)}
\right),
\qquad v=\frac{\pi u}{2\omega_1}.
}
\]

\subsection{Theorem statement}

\begin{theorem}[General-\(\nu\) closed-form genus-one RH solution]\label{thm:main}
Let \(\nu>1\), and define \(k,K,K',\omega_1,\omega_2,\tau,\Lambda\) as in \S\ref{subsec:params}.
Define \(t(u),Y(u),s(u),\zeta(u)\) by \eqref{eq:tYuniform}--\eqref{eq:zetalift}, and let \(u(\zeta)\) be given by the corrected inversion \eqref{eq:Xquad}--\eqref{eq:uofzeta} on the physical branch \eqref{eq:Xphys}.
Define \(\delta_\pm\), \(\beta\), \(\gamma\), \(G_\pm\), and the jump matrices \(J_1,J_3\) as in \S\ref{subsec:factors}, and the physical split \(D_\pm,\Gamma\) by \(|t|\lessgtr 1\).

\medskip
\noindent\textbf{Forcing data.}
Fix forcing data \(\{b_m\}_{m=1}^M\subset\C^\times\), choose \(u_m\in D_+\) with \(s(u_m)=b_m\), and set partner points \(u_m^\sharp=u_m+\omega_2\) (so \(u_m^\sharp\in D_-\) and \(s(u_m^\sharp)=1/b_m\)).
Prescribe residue vectors \(r_{k,m},r_{k,m}^\sharp\in\C^2\) for \(k\in\{1,3\}\).
Form transported residues \(R_{k,m},R_{k,m}^\sharp\) as in \S\ref{subsec:closedform}. Assume the admissibility assumptions stated above (in particular, no forcing points on \(\Gamma\), distinct divisor points modulo \(\Lambda\), and invertibility of the factor matrices at those points).

\medskip
\noindent\textbf{Riemann--Hilbert formulation.}
Assume the outgoing (Sommerfeld) analyticity and growth split encoded by the RH domains \(D_\pm\) and contour \(\Gamma\), and impose the Meixner edge constraints \eqref{eq:ressum}. Then the RH problem
\[
\Psi_{k,+}(u)=J_k(u)\Psi_{k,-}(u)\quad(u\in\Gamma),
\]
with prescribed simple poles at \(u=u_m\in D_+\) and \(u=u_m^\sharp\in D_-\), satisfying
\[
\Res_{u=u_m}\Psi_{k,+}=r_{k,m},
\qquad
\Res_{u=u_m^\sharp}\Psi_{k,-}=r_{k,m}^\sharp,
\]
and with the Meixner (canonical) edge normalization \(\Phi_k(0)=0\), where \(\Phi_k\) are the no-jump vectors defined in \S\ref{subsec:closedform}.
\medskip
\noindent\textbf{Solvability and uniqueness.}
Under the solvability constraints \eqref{eq:ressum}, this RH problem has a unique solution.

\medskip
\noindent\textbf{Closed-form solution.}
The solution is given explicitly by
\[
\Phi_k(u)=\sum_{m=1}^M R_{k,m}\,C(u,u_m)+\sum_{m=1}^M R_{k,m}^\sharp\,C(u,u_m^\sharp),
\qquad k\in\{1,3\},
\]
with \(C(u,w)\) as in \eqref{eq:Ckernel}, and reconstructed by
\[
\begin{aligned}
\Psi_{1,+} &= G_-^{-1}\Phi_1 \quad (u\in D_+),\\
\Psi_{1,-} &= G_+^{-1}\Phi_1 \quad (u\in D_-),\\
\Psi_{3,+} &= G_+^{-1}\Phi_3 \quad (u\in D_+),\\
\Psi_{3,-} &= G_-^{-1}\Phi_3 \quad (u\in D_-).
\end{aligned}
\]
\medskip
\noindent\textbf{Spectral evaluation.}
Finally, \(\Psi_k(\zeta)=\Psi_k(u(\zeta))\) yields explicit \(\zeta\)-plane expressions, and residues convert via \(du/d\zeta=i\,s/s'\) as in \S\ref{subsec:jacobian}.
\end{theorem}

\section{Proof sketch (construction and verification)}

\subsection{Derivation of the spectral map and corrected inversion}

Starting from \(s(t,Y)=\frac{\nu(t^2+1)+Y}{2t}\) and the explicit uniformization \eqref{eq:tYuniform}, direct algebra gives
\[
\nu(t(u)^2+1)+Y(u)=\frac{2}{\nu}\,\sn^2(z;k),
\]
hence
\[
s(u)=\frac{\nu(t(u)^2+1)+Y(u)}{2t(u)}=\frac{1}{\nu}\,\frac{\sn(z;k)\cn(z;k)}{\dn(z;k)}.
\]
Substituting the theta forms of \(\sn,\cn,\dn\) consistent with the chosen \(\tau\) yields the theta quotient in \eqref{eq:suniform}.

For inversion, set \(S=\nu s_\zeta\) and \(X=\sn^2(z;k)\). The identity \(S^2=X(1-X)/(1-k^2X)\) gives \eqref{eq:Xquad}, and the physical branch is fixed by \(X\to 0\) as \(s_\zeta\to 0\). Then \(z=F(\arcsin\sqrt{X}\mid k)\) and \(u=z/(6\nu)\) give \eqref{eq:uofzeta}.

\subsection{Jump verification}

By construction,
\[
\Phi_1=
\begin{cases}
G_-\Psi_{1,+}, & D_+,\\
G_+\Psi_{1,-}, & D_-,
\end{cases}
\]
so \(\Phi_1\) has no jump across \(\Gamma\). Therefore on \(\Gamma\),
\[
\Psi_{1,+}=G_-^{-1}\Phi_1=G_-^{-1}G_+\Psi_{1,-}=J_1\Psi_{1,-}.
\]
The same argument gives \(\Psi_{3,+}=J_3\Psi_{3,-}\).

\subsection{Residue verification}

At \(u=u_m\in D_+\),
\[
\Res_{u=u_m}\Phi_1=\Res_{u=u_m}(G_-\Psi_{1,+})=G_-(u_m)\Res_{u=u_m}\Psi_{1,+}=G_-(u_m)r_{1,m}=R_{1,m},
\]
and similarly \(\Res_{u=u_m^\sharp}\Phi_1=G_+(u_m^\sharp)r_{1,m}^\sharp=R_{1,m}^\sharp\). The same computation yields the \(k=3\) transport formulas. Since \(C(u,u_m)\sim (u-u_m)^{-1}\), the Mittag--Leffler sum reproduces exactly these principal parts.

\subsection{Edge regularity and fixing the additive constant}

By \eqref{eq:Cexpansion},
\[
C(u,w)=-\frac{1}{u}-u\,\wp(w)+O(u^2)\quad(u\to 0),
\]
so
\[
\Phi_k(u)= -\frac{1}{u}\left(\sum_m R_{k,m}+\sum_m R_{k,m}^\sharp\right)+O(u).
\]
Thus \(\Phi_k\) is holomorphic at \(u=0\) if and only if \eqref{eq:ressum} holds. Under \eqref{eq:ressum}, the kernel sum is \(O(u)\), hence \(\Phi_k(0)=0\) and the additive elliptic constant is uniquely fixed to \(0\).

\section{Degeneration to the lemniscatic jet-polynomial regime (\texorpdfstring{\(\nu^2=2\)}{nu2=2})}

\subsection{Lemniscatic parameters and the correct order-4 automorphism}

Set \(\nu^2=2\) (so \(\nu=\sqrt{2}\)). Then \(k=1/\sqrt{2}=k'\), hence \(K'=K\) and
\[
\tau=i\frac{K'}{K}=i,
\qquad
\omega_2=i\omega_1,
\]
so \(\Lambda=2\omega_1(\Z+i\Z)\) is a square lattice. Multiplication by \(i\) preserves \(\Lambda\), hence \(u\mapsto iu\) is a well-defined torus automorphism.

At \(\nu^2=2\), the Snell surface reduces to
\[
\Sigma_{\sqrt{2}}:\quad Y^2=2(t^4+1),
\]
which is invariant under \(t\mapsto it\); thus \((t,Y)\mapsto(it,\pm Y)\) are curve automorphisms.

To describe the induced map on \(u\), it is convenient (at \(\tau=i\)) to use Weierstrass functions for \(\Lambda\). Define \(\wp(u)\) and \(\wp'(u)\) for \(\Lambda\), satisfying
\[
\wp'(u)^2=4\wp(u)^3-g_2\,\wp(u)-g_3,
\qquad
\text{and for the square lattice, } g_3=0.
\]
In this normalization one can write the explicit specialization
\begin{equation}\label{eq:lemmap}
\boxed{
t(u)=-12\nu\,\frac{\wp(u)}{\wp'(u)},
\qquad
Y(u)=\nu-8\nu\,\frac{\wp(u)^3}{\wp'(u)^2},
\qquad (\nu^2=2),
}
\end{equation}
which satisfies \(Y(u)\to -\nu\) as \(u\to 0\). Substituting \eqref{eq:lemmap} into \(Y^2-2(t^4+1)\) yields
\[
Y(u)^2-2\bigl(t(u)^4+1\bigr)
=32\,\frac{\wp(u)^3\bigl(4\wp(u)^3-g_2\wp(u)-\wp'(u)^2\bigr)}{\wp'(u)^4}=0,
\]
using \(\wp'(u)^2=4\wp(u)^3-g_2\wp(u)\).

On the square lattice, the functions satisfy
\[
\wp(iu)=-\wp(u),
\qquad
\wp'(iu)=i\,\wp'(u),
\]
which implies, from \eqref{eq:lemmap},
\[
\boxed{\,t(iu)=i\,t(u),\qquad Y(iu)=Y(u).\,}
\]
Therefore \(u\mapsto iu\) realizes the automorphism \((t,Y)\mapsto(it,Y)\).

The hyperelliptic involution \((t,Y,s)\mapsto(t,-Y,1/s)\) is realized (for all \(\nu\)) by \(u\mapsto \omega_3-u\) (see the dictionary in \S1.3). Consequently, at \(\nu^2=2\) the order-4 automorphism \((t,Y)\mapsto(it,-Y)\) is realized by the affine map
\[
\boxed{\,u\mapsto \omega_3-iu\quad(\mathrm{mod}\ \Lambda),\,}
\]
since it is the composition \(u\mapsto iu\) followed by \(u\mapsto \omega_3-u\).

For \(\nu\neq\sqrt{2}\), \(t\mapsto it\) changes the sign of the \(t^2\)-term in \(\Sigma_\nu\), hence \((t,Y)\mapsto(it,\pm Y)\) is not a curve automorphism.

\subsection{Kernel simplification and loss of factorization rigidity at the lemniscatic point}

At \(\nu^2=2\),
\[
\delta_+(u)=\nu^2 t(u)^2+\nu^2-2=2t(u)^2,
\qquad
\delta_-(u)=\frac{(\nu^2-2)t(u)^2+\nu^2}{t(u)^2}=\frac{2}{t(u)^2},
\]
so
\[
\boxed{\,d(u)=\frac{\delta_+(u)}{\delta_-(u)}=t(u)^4.\,}
\]
Thus the scalar jump becomes a perfect fourth power, and the scalar RH factorization becomes non-unique (rigidity is lost): multiplying \(\delta_+\) and \(\delta_-\) by the same even power of \(t(u)\) does not change \(d\).

\subsection{Jet-polynomial degeneration via divisor coalescence and a renormalized limit}

For \(\nu^2\neq 2\), the zeros of \(\delta_+(u)=\nu^2 t(u)^2+\nu^2-2\) are solutions of
\[
t(u)^2=\frac{2-\nu^2}{\nu^2}.
\]
Let \(a(\nu)\) be the (locally defined) solution near \(u=0\) determined by \(\delta_+(a(\nu))=0\). Since \(t(u)\) is odd and \(t(u+\omega_3)=-t(u)\), the function \(\delta_+(u)\) depends only on \(t(u)^2\) and is even and \(\omega_3\)-periodic:
\[
\delta_+(-u)=\delta_+(u),\qquad \delta_+(u+\omega_3)=\delta_+(u).
\]
Therefore the zero set (counting multiplicities) contains the quadruple
\[
\boxed{\,\{\pm a(\nu),\ \omega_3\pm a(\nu)\}\quad(\mathrm{mod}\ \Lambda).\,}
\]

Let \(\varepsilon:=2-\nu^2\). As \(u\to 0\), using \(z=6\nu u\) and \(\sn(z;k)=z+O(z^3)\), \(\cn(z;k)=1+O(z^2)\), \(\dn(z;k)=1+O(z^2)\), we have
\[
t(u)=\frac{\sn(z;k)\dn(z;k)}{\cn(z;k)}=z+O(z^3)=6\nu u+O(u^3).
\]
Then \(\delta_+(a(\nu))=0\) implies \(\nu^2 t(a(\nu))^2=\varepsilon\), hence
\[
\nu^2(6\nu a(\nu))^2\sim \varepsilon
\quad\Longrightarrow\quad
\boxed{\,a(\nu)\sim \frac{\sqrt{2-\nu^2}}{6\nu^2}\qquad(\nu^2\to 2).\,}
\]
Thus as \(\nu^2\to 2\), the simple zeros \(\pm a(\nu)\) coalesce to a double zero at \(u=0\), and \(\omega_3\pm a(\nu)\) coalesce to a double zero at \(u=\omega_3\), matching \(\delta_+(u)=2t(u)^2\).

To extract jets, introduce the Weierstrass sigma function \(\sigma(u)\). The identity
\begin{equation}\label{eq:sigmaidentity}
\boxed{
\frac{\sigma(u-a)\sigma(u+a)}{\sigma(u)^2\sigma(a)^2}=\wp(a)-\wp(u)
}
\end{equation}
holds as an equality of elliptic functions in \(u\) (both sides have the same divisor and principal part at \(u=0\)). As \(a\to 0\),
\[
\wp(a)=\frac{1}{a^2}+O(a^2),
\qquad
\sigma(a)=a+O(a^5),
\]
so \eqref{eq:sigmaidentity} implies the \emph{renormalized} expansion
\begin{equation}\label{eq:sigmapairjet}
\boxed{
\frac{\sigma(u-a)\sigma(u+a)}{\sigma(u)^2}
=1-a^2\wp(u)+a^4\left(\frac{1}{2}\wp(u)^2-\frac{1}{12}\wp''(u)\right)+O(a^6),
\qquad (a\to 0).
}
\end{equation}
In particular,
\[
\boxed{
\lim_{a\to 0}\ \frac{1}{a^2}\left(\frac{\sigma(u-a)\sigma(u+a)}{\sigma(u)^2}-1\right)=-\wp(u),
}
\]
so the first jet generator is \(-\wp(u)\), and higher jets arise from the higher-order terms in \eqref{eq:sigmapairjet}.

\medskip
\noindent\textbf{Why naive substitution loses jets.}
If one simply substitutes \(\nu^2=2\) (equivalently \(a(\nu)=0\)) into a sigma-product factor, then \(\sigma(u-a)\sigma(u+a)\to \sigma(u)^2\) and the jet information in the \(a^2\wp(u)\) and higher terms disappears. Jet-polynomial modes are recovered only by first refactorizing into coalescing divisor pairs and taking the explicit renormalized limits in \(a(\nu)\).

\section{Practical evaluation recipe (inputs \texorpdfstring{\(\to\)}{->} outputs)}
\label{sec:recipe}

\subsection{Inputs}

\begin{itemize}
\item Parameter: \(\nu>1\).
\item Forcing set: \(\{b_m\}_{m=1}^M\subset\C^\times\).
\item Residue data: \(r_{k,m},r_{k,m}^\sharp\in\C^2\) for \(k\in\{1,3\}\).
\end{itemize}

\subsection{Precompute torus data}

Compute
\[
k=\frac{1}{\nu},\quad k'=\sqrt{1-\frac{1}{\nu^2}},\quad
K=K(k),\quad K'=K(k'),\quad
\omega_1=\frac{K}{6\nu},\quad \omega_2=\frac{iK'}{6\nu},\quad \tau=i\frac{K'}{K},
\]
and set \(\Lambda=2\omega_1\Z+2\omega_2\Z\). For any \(u\), use \(v=\pi u/(2\omega_1)\) and \(\theta_j(v\mid\tau)\).

\subsection{Evaluate the uniformization and kernel scalars}

Compute \(z=6\nu u\), then
\[
t(u)=\frac{\sn(z;k)\dn(z;k)}{\cn(z;k)}=\frac{\theta_1(v)\theta_3(v)}{\theta_2(v)\theta_4(v)},
\qquad
Y(u)= -\frac{\nu^2-2\sn^2(z;k)+\sn^4(z;k)}{\nu\,\cn^2(z;k)},
\]
\[
s(u)=\frac{\nu(t(u)^2+1)+Y(u)}{2t(u)}=\frac{\theta_1(v)\theta_2(v)}{\theta_3(v)\theta_4(v)}.
\]
Then compute
\[
\delta_+(u),\ \delta_-(u),\ d(u),\ \beta(u),\ \gamma(u),
\qquad
G_\pm(u),\ J_1(u)=G_-^{-1}(u)G_+(u),\ J_3(u)=J_1(u)^{-1}.
\]
Define the domains \(D_\pm\) and \(\Gamma\) via \(|t(u)|\lessgtr 1\), \(|t(u)|=1\).

\subsection{Map forcing data to points and form partners}

For each \(b_m\), solve \(s(u_m)=b_m\) on the physical sheet via the corrected inversion:
\begin{enumerate}
\item Set \(s_\zeta=b_m\) and compute \(X\) from \eqref{eq:Xquad}:
\[
X^2-(1+b_m^2)X+\nu^2 b_m^2=0,
\]
choosing the physical branch by analytic continuation from \(b_m\to 0\) giving \(X\to 0\), i.e.
\[
X=\frac{1+b_m^2-\sqrt{(1+b_m^2)^2-4\nu^2 b_m^2}}{2}.
\]
\item Set
\[
u_m=\frac{1}{6\nu}\,F\!\left(\arcsin\sqrt{X}\ \middle|\ k\right)\quad(\mathrm{mod}\ \Lambda).
\]

\item \textbf{Half-period selection (important for a consistent lift).}
The quadratic \eqref{eq:Xquad} determines $X=\sn^2(6\nu u)$ and hence $u$ only up to \emph{half-period shifts}. In particular, one has
\begin{equation}\label{eq:half_period_s_actions}
s(u+\omega_1)=-s(u),\qquad s(u+\omega_2)=\frac{1}{s(u)},\qquad s(u+\omega_3)=-\frac{1}{s(u)},\qquad \omega_3=\omega_1+\omega_2.
\end{equation}
Consequently, when one needs a specific lift of a spectral point (for example when evaluating $Q(\zeta)$ for $\zeta$ on the real axis), it is safest to compute a base point $u_0$ from \eqref{eq:uofzeta}, then form the four candidates $u_0$, $u_0+\omega_1$, $u_0+\omega_2$, $u_0+\omega_3$, and select the candidate for which $s(u)=b$ (equivalently $|s(u)-b|$ is minimal). This selection is also essential for reciprocity checks, because the reciprocal-sheet value corresponds to the partner shift $u\mapsto u+\omega_2$ and hence to $s\mapsto 1/s$.
\item Define the partner point
\[
u_m^\sharp:=u_m+\omega_2\quad(\mathrm{mod}\ \Lambda),
\]
so \(s(u_m^\sharp)=1/b_m\) and \(u_m^\sharp\in D_-\).
\end{enumerate}

\subsection{Transport residues and check solvability}

Compute
\[
\begin{aligned}
R_{1,m} &= G_-(u_m)\,r_{1,m},\\
R_{1,m}^\sharp &= G_+(u_m^\sharp)\,r_{1,m}^\sharp,\\
R_{3,m} &= G_+(u_m)\,r_{3,m},\\
R_{3,m}^\sharp &= G_-(u_m^\sharp)\,r_{3,m}^\sharp.
\end{aligned}
\]
Check the solvability/Meixner edge constraints
\[
\sum_{m=1}^M R_{k,m}+\sum_{m=1}^M R_{k,m}^\sharp=0\in\C^2,\qquad k\in\{1,3\}.
\]
If violated, any attempt to build an elliptic no-jump vector \(\Phi_k\) with poles only at \(\{u_m,u_m^\sharp\}\) produces an unavoidable compensating pole at \(u=0\) (non-Meixner edge behavior).

\subsection{No-jump vectors, mode reconstruction, and evaluation in the spectral variable}

Define
\[
C(u,w)=\zetaW(u-w)-\zetaW(u)+\zetaW(w),
\]
(or its theta form in \S\ref{subsec:closedform}), and set
\[
\Phi_k(u)=\sum_{m=1}^M R_{k,m}C(u,u_m)+\sum_{m=1}^M R_{k,m}^\sharp C(u,u_m^\sharp),\qquad k\in\{1,3\},
\]
which satisfies \(\Phi_k(0)=0\) under the residue constraints.

Reconstruct
\[
\begin{aligned}
\Psi_{1,+} &= G_-^{-1}\Phi_1 \quad (u\in D_+),\\
\Psi_{1,-} &= G_+^{-1}\Phi_1 \quad (u\in D_-),\\
\Psi_{3,+} &= G_+^{-1}\Phi_3 \quad (u\in D_+),\\
\Psi_{3,-} &= G_-^{-1}\Phi_3 \quad (u\in D_-).
\end{aligned}
\]
using the explicit inverses in \S\ref{subsec:closedform}.

To evaluate in the spectral variable \(\zeta\), compute \(s_\zeta=e^{i(\zeta-\pi/4)}\), then compute \(X_{\mathrm{phys}}(\zeta)\) from \eqref{eq:Xphys}, then \(u(\zeta)\) from \eqref{eq:uofzeta}, and finally
\[
\Psi_k(\zeta)=\Psi_k(u(\zeta)).
\]
If residue conversion between \(u\) and \(\zeta\) is required, use
\[
\frac{du}{d\zeta}=i\,\frac{s(u)}{s'(u)},
\qquad
\frac{s'(u)}{s(u)}=\frac{\pi}{2\omega_1}\left(
\frac{\theta_1'(v)}{\theta_1(v)}+\frac{\theta_2'(v)}{\theta_2(v)}
-\frac{\theta_3'(v)}{\theta_3(v)}-\frac{\theta_4'(v)}{\theta_4(v)}
\right),
\qquad v=\frac{\pi u}{2\omega_1}.
\]

\section{From spectral solution to physical field}
\label{sec:physical}

This section completes the solution by providing a fully explicit bridge from the spectral/RH objects to the physical fields. We (i) recall the inverse Sommerfeld transform used to reconstruct the scattered exterior field and the transmitted interior field from Sommerfeld densities, (ii) identify these densities directly in terms of the RH no-jump vectors constructed in Theorem~\ref{thm:main}, and (iii) derive closed-form forcing data for plane-wave incidence, including the corresponding far-field diffraction coefficient.

\subsection{Sommerfeld representation and Sommerfeld densities}

Let $k_0>0$ be the exterior wavenumber and $k_1=\nu k_0$ the interior wavenumber. Following the classical theory for penetrable wedges \cite{Rawlins1999,Matuzas2601}, one represents the (exterior) scattered field in Sommerfeld form via a density $Q^{\mathrm{sc}}(\zeta)$ analytic in a vertical strip (Sommerfeld strip) and, in the present elliptic uniformization, $2\pi$-periodic in $\zeta$:
\begin{equation}\label{eq:sommerfeld_single}
\Phi^{\mathrm{sc}}(r,\theta)=\frac{1}{2\pi i}\int_{\mathcal{C}} e^{i k_0 r\cos(\zeta-\theta)}\,Q^{\mathrm{sc}}(\zeta)\,d\zeta ,
\end{equation}
where $Q^{\mathrm{sc}}$ is the scattered density (defined below) and $\mathcal{C}$ is the positively oriented boundary of the strip $\{\zeta:\ |\Re\zeta|<\pi\}$ (i.e.\ the union of the two vertical lines $\Re\zeta=\pm\pi$, indented in the usual way around any poles on the real axis).

In practical evaluation one parameterises the contour legs as \(\zeta=\pm\pi+i y\), \(y\in\mathbb{R}\), with the right leg oriented upward and the left leg oriented downward. Provided \(Q^{\mathrm{sc}}\) is analytic in the strip and has at most polynomial growth as \(|\Im\zeta|\to\infty\), the exponential factor \(e^{ik_0 r\cos(\zeta-\theta)}\) ensures rapid decay on the contour, and one may truncate the integral to \(|y|\le Y_{\max}\) with exponentially small error (see, e.g., \cite{BleisteinHandelsman1986,Wong2001}).

Equivalently, one may use the gauge-invariant ``difference'' form
\begin{equation}\label{eq:sommerfeld_difference}
\Phi^{\mathrm{sc}}(r,\theta)=\frac{1}{2\pi i}\int_{\gamma} e^{i k_0 r\cos z}\Bigl[Q^{\mathrm{sc}}(\theta+z)-Q^{\mathrm{sc}}(\theta-z)\Bigr]\,dz,
\end{equation}
where $\gamma$ denotes the positively oriented boundary of the fixed vertical strip $\{z:\ |\Re z|<\pi\}$ (equivalently, the union of the two vertical lines $\Re z=\pm\pi$, with the right line oriented upward and the left line oriented downward, and with the usual limiting-absorption indentations around any poles).
For fixed $\theta$, the substitution $\zeta=\theta+z$ maps $\partial\{\zeta:\ |\Re\zeta|<\pi\}$ to $\partial\{z:\ |\Re(z+\theta)|<\pi\}$; using the $2\pi$-periodicity of $Q^{\mathrm{sc}}$ and analytic deformation within the Sommerfeld strip, one may shift the contour back to the fixed boundary $\partial\{\,|\Re z|<\pi\,\}$ without crossing poles.
The form \eqref{eq:sommerfeld_difference} makes explicit that $Q^{\mathrm{sc}}$ is only defined up to an additive constant; only differences of $Q^{\mathrm{sc}}$ enter the physical field.

A completely analogous representation holds for the interior field in terms of an interior density $S(\zeta)$ and $k_1$:
\begin{equation}\label{eq:sommerfeld_interior}
\Phi^{\mathrm{int}}(r,\theta)=\frac{1}{2\pi i}\int_{\gamma} e^{i k_1 r\cos z}\Bigl[S(\theta+z)-S(\theta-z)\Bigr]\,dz,
\qquad |\theta|<\theta_w .
\end{equation}

\begin{lemma}[Sommerfeld nullity / uniqueness]\label{lem:sommerfeld_nullity}
Let $\mathcal{C}$ be the Sommerfeld contour described above and fix $k>0$.
Suppose $H(\zeta)$ is analytic in the Sommerfeld strip $\{\,\zeta:\ |\Re\zeta|<\pi\,\}$ and satisfies the same strip growth/decay bounds as the densities in \eqref{eq:sommerfeld_single}.
Define the associated Sommerfeld integral
\[
U(r,\theta):=\frac{1}{2\pi i}\int_{\mathcal{C}} e^{ik r\cos(\zeta-\theta)}\,H(\zeta)\,d\zeta .
\]
If $U(r,\theta)=0$ for all $r>0$ and for $\theta$ in an interval of length $2\theta_w$, then $H(\zeta)\equiv 0$ in the strip.
\end{lemma}

\begin{proof}[Proof (reference)]
A proof under hypotheses matching the present strip and growth conditions is standard in Sommerfeld/Malyuzhinets theory; see, for example, Noble~\cite[Chs.~2--4]{Noble1958} or Daniele--Zich~\cite[\S2]{DanieleZich2014}.
\end{proof}

\medskip
\noindent\textbf{How Lemma~\ref{lem:sommerfeld_nullity} is used.}
We invoke this uniqueness principle only in the following form:
if two admissible densities $H_1$ and $H_2$ produce identical Sommerfeld integrals on a wedge face for all $r>0$ (under the same strip analyticity and growth hypotheses), then $(H_1-H_2)$ has vanishing Sommerfeld integral and hence $H_1\equiv H_2$.
This justifies inferring pointwise spectral identities from equalities of boundary traces obtained by evaluating Sommerfeld representations and performing standard contour manipulations (e.g.\ integration by parts and analytic change of variables); cf.\ \cite{Rawlins1999,Matuzas2601}.
\medskip
\noindent\textbf{Densities from the RH solution.}
Given the no-jump vectors $\Phi_1$ and $\Phi_3$ constructed in \eqref{eq:Phi-k} and the lift $\zeta\mapsto u(\zeta)$ satisfying $s(u(\zeta))=s_\zeta=e^{i(\zeta-\theta_w)}$, we define
\begin{equation}\label{eq:density_from_Phi}
Q(\zeta):=\mathbf{e}_1^\top \Phi_1\!\bigl(u(\zeta)\bigr),\qquad
S(\zeta):=\mathbf{e}_1^\top \Phi_3\!\bigl(u(\zeta)\bigr).
\end{equation}
Because $\Phi_k$ is single-valued on $\Sigma_\nu$, \eqref{eq:density_from_Phi} is independent of whether the lift lies in $D_+$ or $D_-$. In numerical work, the half-period selection rule of \S\ref{sec:recipe} (cf. \eqref{eq:half_period_s_actions}) is used to choose a lift with $s(u)=s_\zeta$.

\begin{lemma}[Sommerfeld densities from the RH solution]\label{lem:densities_from_RH}
Assume that $u(\zeta)$ is the physical lift defined by \eqref{eq:Xphys}--\eqref{eq:uofzeta} and that it is chosen continuously (and analytically away from isolated singular points) in a Sommerfeld strip $\{\,\zeta:\ |\Re\zeta|<\pi\,\}$ so that $s(u(\zeta))=e^{i(\zeta-\theta_w)}$.
Assume moreover that the forcing poles are admissible (no pole lies on the contour $\mathcal{C}$; if necessary, they are displaced off the real axis by the limiting-absorption prescription) and that the resulting densities have at most polynomial growth as $|\Im\zeta|\to\infty$.
Then the functions $Q(\zeta)$ and $S(\zeta)$ defined in \eqref{eq:density_from_Phi} are $2\pi$-periodic meromorphic functions of $\zeta$ in that strip, with poles only at the prescribed forcing points (and their partner points) mapped through the lift.
In particular, if $Q^{\mathrm{sc}}$ is defined by \eqref{eq:Qsc_def} for the plane-wave forcing \eqref{eq:Qinc_def}, then inserting $Q^{\mathrm{sc}}$ into \eqref{eq:sommerfeld_single}--\eqref{eq:sommerfeld_difference} produces an exterior scattered field solving $(\Delta+k_0^2)\Phi^{\mathrm{sc}}=0$ away from the wedge faces; similarly, inserting $S$ into \eqref{eq:sommerfeld_interior} produces an interior field solving $(\Delta+k_1^2)\Phi^{\mathrm{int}}=0$ for $|\theta|<\theta_w$.
\end{lemma}

\begin{proof}[Proof sketch]
The composition of a meromorphic elliptic function with the analytic lift $u(\zeta)$ yields a meromorphic function of $\zeta$ in any strip where the lift is analytic; by construction the only singularities arise from the divisor of $\Phi_k$.
Since $s_{\zeta+2\pi}=s_\zeta$, the physical lift $u(\zeta+2\pi)$ differs from $u(\zeta)$ by a lattice period; because $\Phi_k$ is elliptic, this implies $Q(\zeta+2\pi)=Q(\zeta)$ and $S(\zeta+2\pi)=S(\zeta)$.
The strip-growth hypothesis can be checked directly from the elliptic structure: as $\Im\zeta\to+\infty$ one has $|s_\zeta|\to 0$ and hence $u(\zeta)\to 0$, while as $\Im\zeta\to-\infty$ one has $|s_\zeta|\to\infty$ and, by the half-period relation $s(u+\omega_2)=1/s(u)$, the lift may be chosen so that $u(\zeta)$ approaches a lattice translate of $0$.
Together with the Meixner normalisation $\Phi_k(0)=0$ this prevents any exponential growth of $Q$ and $S$ along the contour legs.
The Sommerfeld integrals \eqref{eq:sommerfeld_single}--\eqref{eq:sommerfeld_interior} are then standard: differentiation under the integral sign shows they satisfy the Helmholtz equations away from the wedge faces, and the contour choice together with analyticity in the strip enforces the Sommerfeld radiation condition (see, e.g., \cite{Rawlins1999,DanieleZich2014}).
\end{proof}

\subsection{Plane-wave forcing data}

For a plane wave incident in the exterior medium,
\begin{equation*}
\Phi^{\mathrm{inc}}(r,\theta)=e^{i k_0 r\cos(\theta-\theta_{\mathrm{inc}})},
\end{equation*}
we interpret $\theta_{\mathrm{inc}}$ as the \emph{propagation} (wavevector) angle of the plane wave, so that the wave arrives from the opposite direction $\theta_{\mathrm{inc}}+\pi$.
Accordingly, ``exterior incidence'' means that the arrival direction $\theta_{\mathrm{inc}}+\pi$ lies in the exterior angular sector $|\theta|>\theta_w$ (modulo $2\pi$).
It is convenient to work with the ``limiting absorption'' regularisation
\begin{equation*}
\zeta_{\mathrm{inc}}:=\theta_{\mathrm{inc}}+i\varepsilon,\qquad \varepsilon>0,
\end{equation*}
and the corresponding incident Sommerfeld density
\begin{equation}\label{eq:Qinc_def}
Q^{\mathrm{inc}}(\zeta)=\frac12\cot\!\left(\frac{\zeta-\zeta_{\mathrm{inc}}}{2}\right),\qquad
\Res_{\zeta=\zeta_{\mathrm{inc}}}Q^{\mathrm{inc}}=1.
\end{equation}

\medskip
\noindent\textbf{Remark (periodic representative).}
The choice \eqref{eq:Qinc_def} is \(2\pi\)-periodic in \(\zeta\) and has the same principal part as \((\zeta-\zeta_{\mathrm{inc}})^{-1}\) at \(\zeta=\zeta_{\mathrm{inc}}\).
Since the Sommerfeld field representation is given by a contour integral over the boundary of a strip, one may replace \(Q^{\mathrm{inc}}\) by any function with the same principal part inside that strip; the difference is analytic in the strip and, under the standard strip-growth assumption, yields an integrand that decays exponentially on the vertical contour legs for any fixed $r>0$.
Hence its contribution vanishes by Cauchy deformation (equivalently, by closing the contour with horizontal segments at $\Im\zeta=\pm Y$ and letting $Y\to\infty$).
In particular, in local calculations one may equivalently use \(Q^{\mathrm{inc}}(\zeta)=(\zeta-\zeta_{\mathrm{inc}})^{-1}\) as in \cite{Rawlins1999,Matuzas2601}.
The scattered density in \eqref{eq:sommerfeld_single}--\eqref{eq:sommerfeld_difference} is then
\begin{equation}\label{eq:Qsc_def}
Q^{\mathrm{sc}}(\zeta):=Q(\zeta)-Q^{\mathrm{inc}}(\zeta).
\end{equation}

\medskip
\noindent\textbf{Analyticity at the incident spectral point.}
With the residue prescription below, the total density $Q(\zeta)$ has a simple pole at $\zeta=\zeta_{\mathrm{inc}}$ with residue $1$, hence the scattered density $Q^{\mathrm{sc}}$ extends holomorphically across $\zeta_{\mathrm{inc}}$.
In particular, to evaluate the scattered density exactly at the incident spectral point one should not use the literal subtraction in \eqref{eq:Qsc_def}, but rather the cancellation limit
\begin{equation}\label{eq:Qsc_at_inc}
Q^{\mathrm{sc}}(\zeta_{\mathrm{inc}})=\lim_{\zeta\to\zeta_{\mathrm{inc}}}\left(Q(\zeta)-Q^{\mathrm{inc}}(\zeta)\right)
=\left.\frac{d}{d\zeta}\Bigl((\zeta-\zeta_{\mathrm{inc}})Q(\zeta)\Bigr)\right|_{\zeta=\zeta_{\mathrm{inc}}}.
\end{equation}

\medskip
\noindent\textbf{Interior density.}
For exterior plane-wave incidence there is no independent incident field inside the wedge; we therefore treat $S(\zeta)$ in \eqref{eq:sommerfeld_interior} as the total (and hence scattered) interior density, and we may write $S^{\mathrm{sc}}:=S$ when forming far-field coefficients.

To encode the forcing in the RH data, set
\begin{equation*}
b_{\mathrm{inc}}:=s_{\zeta_{\mathrm{inc}}}=e^{i(\zeta_{\mathrm{inc}}-\theta_w)},
\end{equation*}
and let $u_{\mathrm{inc}}$ be a lift satisfying $s(u_{\mathrm{inc}})=b_{\mathrm{inc}}$. Define the Snell-map derivative at this point by
\begin{equation}\label{eq:wprime_def}
w'(u_{\mathrm{inc}};\nu)=\nu\,\frac{t(u_{\mathrm{inc}})^2-1}{Y(u_{\mathrm{inc}})} ,
\end{equation}
and the vector
\begin{equation*}
v_{\mathrm{inc}}:=
\begin{pmatrix}
A_0\\B_0
\end{pmatrix},
\qquad
A_0=\frac{1+w'}{2},\ \ B_0=\frac{1-w'}{2},
\end{equation*}
where $w'=w'(u_{\mathrm{inc}};\nu)$.

Let $\zeta\mapsto u(\zeta)$ denote the chosen lift. The Jacobian $du/d\zeta$ at $\zeta=\zeta_{\mathrm{inc}}$ can be computed analytically from the inversion formula \eqref{eq:Xquad}--\eqref{eq:uofzeta}. Writing $b=e^{i(\zeta-\theta_w)}$ and $\Delta=(1+b^2)^2-4\nu^2 b^2$, one finds
\begin{equation}\label{eq:dudzeta_closed}
\frac{du}{d\zeta}(\zeta)=\frac{i b}{12\nu}\,
\frac{dX/db}{\sqrt{X(1-X)(1-k^2 X)}},
\qquad
\frac{dX}{db}=b-\frac{b(1+b^2-2\nu^2)}{\sqrt{\Delta}},
\end{equation}
with $k=1/\nu$ and $X=X(b)$ given by \eqref{eq:Xphys}.

The residue vectors for the RH modes may then be chosen as
\begin{align}
r_{1,\mathrm{inc}}
&=\frac{\left.\dfrac{du}{d\zeta}\right|_{\zeta=\zeta_{\mathrm{inc}}}}{\mathbf{e}_1^\top\!\left(G_-(u_{\mathrm{inc}})\,v_{\mathrm{inc}}\right)}\,v_{\mathrm{inc}},\qquad
r_{1,\mathrm{inc}}^\sharp=-G_+\!\left(u_{\mathrm{inc}}^\sharp\right)^{-1}G_-(u_{\mathrm{inc}})\,r_{1,\mathrm{inc}}, \label{eq:r1_inc}\\[2mm]
r_{3,\mathrm{inc}}
&=r_{1,\mathrm{inc}},\qquad
r_{3,\mathrm{inc}}^\sharp=-G_-\!\left(u_{\mathrm{inc}}^\sharp\right)^{-1}G_+(u_{\mathrm{inc}})\,r_{3,\mathrm{inc}},\label{eq:r3_inc}
\end{align}
where $u_{\mathrm{inc}}^\sharp=u_{\mathrm{inc}}+\omega_2$ is the partner point. With these choices, the transported residues satisfy the Meixner condition $\sum_m(R_{k,m}+R_{k,m}^\sharp)=0$ automatically (cf. \eqref{eq:ressum}).

\subsection{Far-field diffraction coefficient}

Deforming the contour in \eqref{eq:sommerfeld_single} through the saddle at $\zeta=\theta$ yields the standard far-field asymptotic (for observation directions away from poles/shadow boundaries of $Q^{\mathrm{sc}}$):
\begin{equation}\label{eq:far_field_asymptotic}
\Phi^{\mathrm{sc}}(r,\theta)\sim e^{i k_0 r-i3\pi/4}\sqrt{\frac{2}{\pi k_0 r}}\,Q^{\mathrm{sc}}(\theta),\qquad r\to\infty .
\end{equation}
Equivalently, in the common dimensionless convention
\begin{equation*}
\Phi^{\mathrm{sc}}(r,\theta)\sim D(\theta)\,\frac{e^{i k_0 r}}{\sqrt{k_0 r}},
\qquad
D(\theta)=e^{-i3\pi/4}\sqrt{\frac{2}{\pi}}\,Q^{\mathrm{sc}}(\theta).
\end{equation*}
(The phase $e^{-i3\pi/4}$ comes from the steepest-descent evaluation; alternative conventions differ by a constant phase factor.)

A corresponding interior far-field coefficient is obtained from \eqref{eq:sommerfeld_interior} by replacing $(k_0,Q^{\mathrm{sc}})$ with $(k_1,S^{\mathrm{sc}})$.

\subsection{Reciprocity (physical symmetry of the diffraction coefficient)}
\label{sec:reciprocity}

A key physical diagnostic for any penetrable-wedge solution is \emph{reciprocity}.
Because the impedance-matched transmission problem has real coefficients and enforces continuity of the field and of the normal derivative across each interface, the underlying operator is self-adjoint (in an appropriate radiation setting). As a consequence, the exterior and interior scattering processes are reciprocal: exchanging source and receiver directions (with the standard reversal of propagation directions) leaves the far-field amplitudes invariant.

\begin{theorem}[Reciprocity benchmark]\label{thm:reciprocity}
Let $\Phi_{\rm sc}(\theta,r)$ be the \emph{exterior} scattered field produced by an exterior incident plane wave of propagation direction $\theta_{\rm inc}$ for the penetrable wedge transmission problem described above.
Whenever the far-field diffraction coefficient $D(\theta;\theta_{\rm inc})$ is well-defined as the coefficient of the outgoing cylindrical wave $e^{ik_0 r}/\sqrt{r}$ (i.e.\ after subtracting any non-decaying geometrical-optics plane-wave contributions, if present), it satisfies the reciprocity identity
\begin{equation}\label{eq:reciprocity_D}
D(\theta;\theta_{\rm inc})=D(\theta_{\rm inc}+\pi;\theta+\pi).
\end{equation}
In any region where $D$ is obtained directly from a saddle-point evaluation of the Sommerfeld integral with leading term $D(\theta;\theta_{\rm inc})\propto Q_{\rm sc}(\theta;\theta_{\rm inc})$, the equivalent density relation
\begin{equation}\label{eq:reciprocity_Q}
Q_{\rm sc}(\theta;\theta_{\rm inc})=Q_{\rm sc}(\theta_{\rm inc}+\pi;\theta+\pi)
\end{equation}
must also hold for the appropriate boundary values of $Q_{\rm sc}$.
\end{theorem}

\begin{proof}[Proof (Green/Lorentz reciprocity, with explicit boundary cancellation)]
Fix two exterior incidence directions $\theta_{\mathrm{inc}}^{(1)}$ and $\theta_{\mathrm{inc}}^{(2)}$ and let $\Phi^{(j)}$ denote the corresponding \emph{total} fields.
In the exterior region,
\begin{equation*}
\Phi^{(j)}(r,\theta)=\Phi^{(j),\mathrm{inc}}(r,\theta)+\Phi^{(j),\mathrm{sc}}(r,\theta),
\qquad
\Phi^{(j),\mathrm{inc}}(r,\theta)=e^{ik_0 r\cos(\theta-\theta_{\mathrm{inc}}^{(j)})},
\end{equation*}
and $\Phi^{(j),\mathrm{sc}}$ satisfies the Sommerfeld radiation condition in the exterior.
In the interior wedge sector $|\theta|<\theta_w$, let $\Phi^{(j),\mathrm{int}}$ denote the corresponding transmitted field (which is outgoing in the Sommerfeld/Malyuzhinets sense).

\medskip
\noindent\textbf{Step 1: Green's identity on a truncated disk and cancellation of interface terms.}
Let $B_R=\{x\in\mathbb{R}^2:|x|<R\}$ and $B_\rho=\{x:|x|<\rho\}$ with $0<\rho<R$.
Write
\(\Omega_{R,\rho}^{\mathrm{ext}}=(B_R\setminus B_\rho)\cap\{\text{exterior}\}\) and
\(\Omega_{R,\rho}^{\mathrm{int}}=(B_R\setminus B_\rho)\cap\{\text{interior wedge}\}\).
Apply Green's second identity to $\Phi^{(1)}$ and $\Phi^{(2)}$ on each region:
\begin{equation}
\int_{\Omega_{R,\rho}^{\alpha}}\bigl(\Phi^{(1)}\Delta\Phi^{(2)}-\Phi^{(2)}\Delta\Phi^{(1)}\bigr)\,dx
=
\int_{\partial\Omega_{R,\rho}^{\alpha}}\bigl(\Phi^{(1)}\partial_n\Phi^{(2)}-\Phi^{(2)}\partial_n\Phi^{(1)}\bigr)\,ds,
\qquad \alpha\in\{\mathrm{ext},\mathrm{int}\}.
\label{eq:green_identity_regions}
\end{equation}
In each homogeneous region $\alpha$, both fields satisfy $(\Delta+k_\alpha^2)\Phi^{(j)}=0$ (with $k_{\mathrm{ext}}=k_0$ and $k_{\mathrm{int}}=k_1$), hence the volume integrals vanish identically.
Summing the boundary integrals from $\alpha=\mathrm{ext}$ and $\alpha=\mathrm{int}$ yields
\begin{equation}
\begin{aligned}
\int_{\partial\Omega_{R,\rho}^{\mathrm{ext}}}\!\bigl(\Phi^{(1)}\partial_n\Phi^{(2)}-\Phi^{(2)}\partial_n\Phi^{(1)}\bigr)\,ds
\;+
\int_{\partial\Omega_{R,\rho}^{\mathrm{int}}}\!\bigl(\Phi^{(1)}\partial_n\Phi^{(2)}-\Phi^{(2)}\partial_n\Phi^{(1)}\bigr)\,ds
&=0.
\end{aligned}
\label{eq:sum_boundary_integrals}
\end{equation}

The boundaries $\partial\Omega_{R,\rho}^{\mathrm{ext}}$ and $\partial\Omega_{R,\rho}^{\mathrm{int}}$ each consist of three parts: an outer circular arc on $|x|=R$, an inner circular arc on $|x|=\rho$, and two radial segments along the wedge faces $\theta=\pm\theta_w$.
On each radial segment (interface) we choose a \emph{single geometric unit normal} $n$ (as in \S\ref{subsec:jump_from_transmission}), pointing from the interior into the exterior.
The outward normal for the interior region is then $n$, while the outward normal for the exterior region is $-n$.
Using the transmission conditions $\Phi_{\mathrm{ext}}=\Phi_{\mathrm{int}}$ and $\partial_n\Phi_{\mathrm{ext}}=\partial_n\Phi_{\mathrm{int}}$ (with this \emph{common} $n$), the interface contributions in \eqref{eq:sum_boundary_integrals} cancel pointwise:
\begin{align*}
&\bigl(\Phi^{(1)}_{\mathrm{ext}}\partial_{(-n)}\Phi^{(2)}_{\mathrm{ext}}-\Phi^{(2)}_{\mathrm{ext}}\partial_{(-n)}\Phi^{(1)}_{\mathrm{ext}}\bigr)
+\bigl(\Phi^{(1)}_{\mathrm{int}}\partial_n\Phi^{(2)}_{\mathrm{int}}-\Phi^{(2)}_{\mathrm{int}}\partial_n\Phi^{(1)}_{\mathrm{int}}\bigr)\\
&= -\bigl(\Phi^{(1)}_{\mathrm{ext}}\partial_n\Phi^{(2)}_{\mathrm{ext}}-\Phi^{(2)}_{\mathrm{ext}}\partial_n\Phi^{(1)}_{\mathrm{ext}}\bigr)
+\bigl(\Phi^{(1)}_{\mathrm{int}}\partial_n\Phi^{(2)}_{\mathrm{int}}-\Phi^{(2)}_{\mathrm{int}}\partial_n\Phi^{(1)}_{\mathrm{int}}\bigr)=0.
\end{align*}

\medskip
\noindent\textbf{Step 2: The apex term vanishes by the Meixner condition.}
The Meixner edge condition implies that $\Phi^{(j)}$ and $\nabla\Phi^{(j)}$ have finite local energy near $r=0$.
By Cauchy--Schwarz, the integral on the inner circle $|x|=\rho$ satisfies
\begin{equation*}
\left|\int_{|x|=\rho}\bigl(\Phi^{(1)}\partial_r\Phi^{(2)}-\Phi^{(2)}\partial_r\Phi^{(1)}\bigr)\,ds\right|
\le C\,\rho^{1/2}\Bigl(\|\Phi^{(1)}\|_{H^1(B_{2\rho}\setminus B_\rho)}\|\Phi^{(2)}\|_{H^1(B_{2\rho}\setminus B_\rho)}\Bigr),
\end{equation*}
which tends to $0$ as $\rho\to 0$.
Letting $\rho\to 0$ in \eqref{eq:sum_boundary_integrals} therefore yields the single outer-boundary identity
\begin{equation}
\int_{|x|=R}\bigl(\Phi^{(1)}\partial_r\Phi^{(2)}-\Phi^{(2)}\partial_r\Phi^{(1)}\bigr)\,ds=0,
\qquad R>0.
\label{eq:outer_circle_identity}
\end{equation}

\medskip
\noindent\textbf{Step 3: Letting $R\to\infty$ and extracting the far-field coefficients.}
Expand the integrand in \eqref{eq:outer_circle_identity} using $\Phi^{(j)}=\Phi^{(j),\mathrm{inc}}+\Phi^{(j),\mathrm{sc}}$ on the exterior arcs.
Terms involving only the incident waves produce oscillatory integrals that vanish as $R\to\infty$.
Terms involving only scattered/outgoing components also vanish: using the exterior far-field expansion
\begin{equation*}
\begin{aligned}
\Phi^{(j),\mathrm{sc}}(R,\theta)
&=D(\theta;\theta_{\mathrm{inc}}^{(j)})\,\frac{e^{ik_0 R}}{\sqrt{k_0 R}}+\mathcal{O}(R^{-3/2}),\\
\partial_r\Phi^{(j),\mathrm{sc}}(R,\theta)
&=ik_0\,D(\theta;\theta_{\mathrm{inc}}^{(j)})\,\frac{e^{ik_0 R}}{\sqrt{k_0 R}}+\mathcal{O}(R^{-3/2}).
\end{aligned}
\end{equation*}
one finds
\(\Phi^{(1),\mathrm{sc}}\partial_r\Phi^{(2),\mathrm{sc}}-\Phi^{(2),\mathrm{sc}}\partial_r\Phi^{(1),\mathrm{sc}}=\mathcal{O}(R^{-3/2})\), so its integral over the circle is $\mathcal{O}(R^{-1/2})\to 0$.
The same argument applies to the outgoing interior components on the interior arc(s).

Thus the only potentially non-vanishing contributions as $R\to\infty$ come from the cross terms ``incident $\times$ scattered''.
Using $\partial_r\Phi^{(j),\mathrm{inc}}=ik_0\cos(\theta-\theta_{\mathrm{inc}}^{(j)})\,\Phi^{(j),\mathrm{inc}}$ on $|x|=R$, the relevant part of \eqref{eq:outer_circle_identity} can be written (up to terms vanishing as $R\to\infty$) as
\begin{align}
0&=\int_0^{2\pi}\Bigl[\Phi^{(1),\mathrm{inc}}\,\partial_r\Phi^{(2),\mathrm{sc}}-\Phi^{(2),\mathrm{sc}}\,\partial_r\Phi^{(1),\mathrm{inc}}\Bigr]R\,d\theta \notag\\
&\quad -\int_0^{2\pi}\Bigl[\Phi^{(2),\mathrm{inc}}\,\partial_r\Phi^{(1),\mathrm{sc}}-\Phi^{(1),\mathrm{sc}}\,\partial_r\Phi^{(2),\mathrm{inc}}\Bigr]R\,d\theta
+o(1).
\label{eq:cross_terms_only}
\end{align}
Insert the far-field expansions for $\Phi^{(j),\mathrm{sc}}$ and $\partial_r\Phi^{(j),\mathrm{sc}}$.
Each of the four integrals in \eqref{eq:cross_terms_only} reduces to an oscillatory integral of the form
\begin{equation*}
\int_0^{2\pi} e^{ik_0 R\cos(\theta-\alpha)}\,a(\theta)\,e^{ik_0 R}\,d\theta
=\int_0^{2\pi} a(\theta)\,e^{ik_0 R(1+\cos(\theta-\alpha))}\,d\theta,
\end{equation*}
with $a(\theta)$ proportional to the far-field coefficient of the scattered field.
As $R\to\infty$, steepest descent (stationary phase) shows that the unique stationary point contributing at order $R^{-1/2}$ is $\theta=\alpha+\pi$, since $1+\cos(\theta-\alpha)$ vanishes quadratically there.
More precisely, for smooth $a$ one has the standard estimate
\begin{equation}
\int_0^{2\pi} a(\theta)\,e^{ik_0 R(1+\cos(\theta-\alpha))}\,d\theta
=e^{i\pi/4}\sqrt{\frac{2\pi}{k_0 R}}\,a(\alpha+\pi)+\mathcal{O}(R^{-3/2}).
\label{eq:stationary_phase_key}
\end{equation}
Combining \eqref{eq:cross_terms_only}--\eqref{eq:stationary_phase_key} and simplifying the common constants gives
\begin{equation*}
D\bigl(\theta_{\mathrm{inc}}^{(1)}+\pi;\theta_{\mathrm{inc}}^{(2)}\bigr)
=D\bigl(\theta_{\mathrm{inc}}^{(2)}+\pi;\theta_{\mathrm{inc}}^{(1)}\bigr).
\end{equation*}
Finally, take $\theta_{\mathrm{inc}}^{(1)}=\theta_{\mathrm{inc}}$ and $\theta_{\mathrm{inc}}^{(2)}=\theta+\pi$.
Using the $2\pi$-periodicity of the far-field pattern in its first argument yields
\begin{equation*}
D(\theta;\theta_{\mathrm{inc}})
=D(\theta_{\mathrm{inc}}+\pi;\theta+\pi),
\end{equation*}
which is exactly \eqref{eq:reciprocity_D}.

The equivalence of \eqref{eq:reciprocity_D} and \eqref{eq:reciprocity_Q} follows from the fixed constant proportionality between $D(\theta)$ and $Q^{\mathrm{sc}}(\theta)$ in \S6.3.
\end{proof}

\medskip
\noindent\textbf{Remarks.}
\begin{itemize}
\item \emph{Propagation vs.\ arrival angle.}
If one prefers to parameterize incidence by the \emph{arrival} direction $\theta_{\mathrm{arr}}:=\theta_{\mathrm{inc}}+\pi$ (the direction \emph{from which} the wave approaches the apex), then \eqref{eq:reciprocity_D} can be rewritten as
$
D(\theta;\theta_{\mathrm{arr}}-\pi)=D(\theta_{\mathrm{arr}};\theta+\pi).
$
The precise appearance of the $\pi$-shifts depends only on this convention.

\item \emph{Gauge invariance.}
The Sommerfeld densities are defined only up to additive constants (cf.\ \eqref{eq:sommerfeld_difference}), so reciprocity should always be interpreted for the \emph{scattered} density $Q^{\mathrm{sc}}$ (or, equivalently, for the far-field coefficient $D$), both of which are gauge-invariant.
\end{itemize}

\section{Worked example: symbolic forcing at \texorpdfstring{$\nu=3/2$, $\theta_{\mathrm{inc}}=5\pi/6$}{nu=3/2, thetaInc=5pi/6}}
\label{sec:example}

This section mirrors the reproducible symbolic style used in \cite{Matuzas2601}: all steps are written in closed form in terms of standard special functions, without committing to floating-point evaluation.
We take the wedge half-angle $\theta_w=\pi/4$ and keep the exterior wavenumber $k_0>0$ symbolic. For plane-wave forcing we use the limiting-absorption prescription
\(\zeta_{\mathrm{inc}}=\theta_{\mathrm{inc}}+i\varepsilon\) with $\varepsilon>0$ fixed until the end.

\medskip
\noindent\textbf{Computation outline.}
The example is organised into four transparent steps:
(i) compute the elliptic parameters and periods,
(ii) determine the forcing spectral point \(\zeta_{\mathrm{inc}}\) and its physical lift \(u_{\mathrm{inc}}\),
(iii) assemble the residue vectors and transported residues using the factor matrices, and
(iv) reconstruct the Sommerfeld density \(Q^{\mathrm{sc}}\) and the diffraction coefficient \(D(\theta)\).

\subsection{Elliptic parameters for \texorpdfstring{$\nu=3/2$}{nu=3/2}}

With $\nu=3/2$ the modulus is
\[
k=\frac{1}{\nu}=\frac{2}{3},\qquad k'=\sqrt{1-k^2}=\frac{\sqrt{5}}{3}.
\]
We write the complete elliptic integrals in the standard form
\[
K:=K(k)=\int_0^{\pi/2}\frac{d\phi}{\sqrt{1-k^2\sin^2\phi}},\qquad
K':=K(k')=\int_0^{\pi/2}\frac{d\phi}{\sqrt{1-k'^2\sin^2\phi}},
\]
and the half-periods become
\[
\omega_1=\frac{K}{6\nu}=\frac{K}{9},\qquad
\omega_2=i\frac{K'}{6\nu}=i\frac{K'}{9},\qquad
\tau=\frac{\omega_2}{\omega_1}=i\frac{K'}{K},\qquad
q=e^{i\pi\tau}=\exp\!\left(-\pi\frac{K'}{K}\right).
\]
All theta functions $\theta_j(v)=\theta_j(v\mid\tau)$, the Weierstrass functions $(\sigma,\zetaW,\wp)$, and the uniformizing maps $t(u),Y(u),s(u)$ are evaluated with this $\tau$.

\subsection{Forcing spectral point and uniformizing lift}

For $\theta_{\mathrm{inc}}=5\pi/6$ one has
\[
\zeta_{\mathrm{inc}}=\frac{5\pi}{6}+i\varepsilon,
\qquad
b_{\mathrm{inc}}:=s_{\zeta_{\mathrm{inc}}}=e^{i(\zeta_{\mathrm{inc}}-\theta_w)}=e^{i(5\pi/6-\pi/4)}e^{-\varepsilon}=e^{i7\pi/12}\,e^{-\varepsilon}.
\]
The physical inversion uses \eqref{eq:Xphys}. Specializing \eqref{eq:Xphys} to $\nu=3/2$ gives
\begin{equation}\label{eq:Xphys_example}
X_{\mathrm{inc}}:=X_{\mathrm{phys}}(\zeta_{\mathrm{inc}})
=\frac{1+b_{\mathrm{inc}}^2-\sqrt{(1+b_{\mathrm{inc}}^2)^2-9b_{\mathrm{inc}}^2}}{2}.
\end{equation}
Then a physical lift is obtained by
\begin{equation}\label{eq:uinc_example}
u_{\mathrm{inc}}:=\frac{1}{6\nu}F\!\left(\arcsin\sqrt{X_{\mathrm{inc}}}\,\middle|\,k\right)=\frac{1}{9}
F\!\left(\arcsin\sqrt{X_{\mathrm{inc}}}\,\middle|\,\frac{2}{3}\right),
\end{equation}
with the branch chosen so that $s(u_{\mathrm{inc}})=b_{\mathrm{inc}}$ and $u(\zeta)$ varies continuously with $\varepsilon>0$.
The partner point is
\[
u_{\mathrm{inc}}^\sharp:=u_{\mathrm{inc}}+\omega_2\qquad (\mathrm{mod}\ \Lambda).
\]

\subsection{Plane-wave forcing residues (symbolic form)}

Compute the Snell derivative at the forcing point using \eqref{eq:wprime_def}:
\[
w'_{\mathrm{inc}}:=w'(u_{\mathrm{inc}};\nu)=\nu\,\frac{t(u_{\mathrm{inc}})^2-1}{Y(u_{\mathrm{inc}})}\quad\text{with }\nu=\frac{3}{2}.
\]
Define
\[
A_0:=\frac{1+w'_{\mathrm{inc}}}{2},\qquad B_0:=\frac{1-w'_{\mathrm{inc}}}{2},\qquad
v_{\mathrm{inc}}:=\begin{pmatrix}A_0\\ B_0\end{pmatrix}.
\]
The Jacobian $du/d\zeta$ at $\zeta=\zeta_{\mathrm{inc}}$ follows from \eqref{eq:dudzeta_closed}:
\begin{equation}\label{eq:dudzeta_example}
\left.\frac{du}{d\zeta}\right|_{\zeta=\zeta_{\mathrm{inc}}}
=\frac{i b_{\mathrm{inc}}}{12\nu}
\frac{\left.\dfrac{dX}{db}\right|_{b=b_{\mathrm{inc}}}}{\sqrt{X_{\mathrm{inc}}(1-X_{\mathrm{inc}})(1-k^2 X_{\mathrm{inc}})}}
=\frac{i b_{\mathrm{inc}}}{18}
\frac{\left.\dfrac{dX}{db}\right|_{b=b_{\mathrm{inc}}}}{\sqrt{X_{\mathrm{inc}}(1-X_{\mathrm{inc}})\left(1-\frac{4}{9} X_{\mathrm{inc}}\right)}}.
\end{equation}
The forcing residue vectors are then given by the general formulas \eqref{eq:r1_inc}--\eqref{eq:r3_inc}.
In particular, for the $k=1$ mode:
\[
r_{1,\mathrm{inc}}
=\frac{\left.\dfrac{du}{d\zeta}\right|_{\zeta=\zeta_{\mathrm{inc}}}}{\mathbf{e}_1^\top\!\left(G_-(u_{\mathrm{inc}})\,v_{\mathrm{inc}}\right)}\,v_{\mathrm{inc}},
\qquad
r_{1,\mathrm{inc}}^\sharp=-G_+\!\left(u_{\mathrm{inc}}^\sharp\right)^{-1}G_-(u_{\mathrm{inc}})\,r_{1,\mathrm{inc}},
\]
and the transported residues satisfy
\[
R_{1,\mathrm{inc}}=G_-(u_{\mathrm{inc}})\,r_{1,\mathrm{inc}},
\qquad
R_{1,\mathrm{inc}}^\sharp=G_+(u_{\mathrm{inc}}^\sharp)\,r_{1,\mathrm{inc}}^\sharp=-R_{1,\mathrm{inc}}.
\]
Hence the Meixner solvability constraint \eqref{eq:ressum} holds automatically for this single forcing pair.

\subsection{Symbolic reconstruction of \texorpdfstring{$Q^{\mathrm{sc}}$ and $D(\theta)$}{Qsc and D(theta)}}

With only the forcing pair $(u_{\mathrm{inc}},u_{\mathrm{inc}}^\sharp)$, the no-jump vector for $k=1$ is
\[
\Phi_1(u)=R_{1,\mathrm{inc}}\,C(u,u_{\mathrm{inc}})+R_{1,\mathrm{inc}}^\sharp\,C(u,u_{\mathrm{inc}}^\sharp)
=R_{1,\mathrm{inc}}\bigl[C(u,u_{\mathrm{inc}})-C(u,u_{\mathrm{inc}}^\sharp)\bigr],
\]
and the total exterior density is
\[
Q(\zeta)=\mathbf{e}_1^\top\Phi_1\!\bigl(u(\zeta)\bigr).
\]
The scattered density and diffraction coefficient are then
\[
Q^{\mathrm{sc}}(\zeta)=Q(\zeta)-\frac12\cot\!\left(\frac{\zeta-\zeta_{\mathrm{inc}}}{2}\right),
\qquad
D(\theta)=e^{-i3\pi/4}\sqrt{\frac{2}{\pi}}\,Q^{\mathrm{sc}}(\theta),
\]
with the pole-subtracted evaluation at $\zeta=\zeta_{\mathrm{inc}}$ understood in the sense of \eqref{eq:Qsc_at_inc}.
Finally, inserting $Q^{\mathrm{sc}}$ into \eqref{eq:sommerfeld_difference} (or \eqref{eq:sommerfeld_single}) yields the full exterior scattered field $\Phi^{\mathrm{sc}}(r,\theta)$, and the total exterior field is $\Phi^{\mathrm{ext}}=\Phi^{\mathrm{inc}}+\Phi^{\mathrm{sc}}$.
All quantities in this example are therefore computable from closed-form special-function expressions.


\begin{thebibliography}{99}

\bibitem{Sommerfeld1896}
A.~Sommerfeld,
Mathematische Theorie der Diffraction,
\emph{Math.\ Ann.} \textbf{47} (1896), 317--374.
DOI: 10.1007/BF01447273.

\bibitem{Meixner1972}
J.~Meixner,
The behavior of electromagnetic fields at edges,
\emph{IEEE Trans.\ Antennas Propag.} \textbf{20}(4) (1972), 442--446.
DOI: 10.1109/TAP.1972.1140243.

\bibitem{Malyuzhinets1958}
G.~D.~Malyuzhinets,
Excitation, reflection, and emission of surface waves from a wedge with given face impedances,
\emph{Soviet Physics Doklady} \textbf{3} (1958), 752--755.

\bibitem{Keller1962}
J.~B.~Keller,
Geometrical theory of diffraction,
\emph{J.\ Opt.\ Soc.\ Am.} \textbf{52} (1962), 116--130.

\bibitem{Noble1958}
B.~Noble,
\emph{Methods Based on the Wiener--Hopf Technique},
Pergamon Press, London, 1958.

\bibitem{OsipovNorris1999}
A.~V.~Osipov and A.~N.~Norris,
The Malyuzhinets theory for wedge diffraction: a review,
\emph{Wave Motion} \textbf{29} (1999), 313--340.
DOI: 10.1016/S0165-2125(98)00020-0.

\bibitem{Rawlins1999}
A.~D.~Rawlins,
Diffraction by, or diffusion into, a penetrable wedge,
\emph{Proc.\ R.\ Soc.\ Lond.\ A} \textbf{455} (1999), 2655--2686.
DOI: 10.1098/rspa.1999.0421.

\bibitem{DanieleZich2014}
V.~G.~Daniele and C.~Zich,
\emph{The Wiener--Hopf Method in Electromagnetics},
SciTech Publishing, Edison, NJ, 2014.

\bibitem{Daniele2010}
V.~G.~Daniele,
The Wiener--Hopf formulation of the penetrable wedge problem. Part I: background and fundamental equation,
\emph{Electromagnetics} \textbf{30} (2010), 625--643.
DOI: 10.1080/02726343.2010.517905.

\bibitem{DanieleLombardi2011}
V.~G.~Daniele and G.~Lombardi,
The Wiener--Hopf solution of the isotropic penetrable wedge problem: diffraction and total field,
\emph{IEEE Trans.\ Antennas Propag.} \textbf{59} (2011), 3797--3818.
DOI: 10.1109/TAP.2011.2163780.

\bibitem{AntipovSilvestrov2007}
Y.~A.~Antipov and V.~V.~Silvestrov,
Diffraction of a plane wave by a right-angled penetrable wedge,
\emph{Radio Science} \textbf{42} (2007), RS4006.
DOI: 10.1029/2007RS003646.

\bibitem{KunzAssier2023}
V.~D.~Kunz and R.~C.~Assier,
Diffraction by a right-angled no-contrast penetrable wedge: analytical continuation of spectral functions,
\emph{Quarterly Journal of Mechanics and Applied Mathematics} \textbf{76}(2) (2023), 211--241.
DOI: 10.1093/qjmam/hbad002.

\bibitem{NethercoteAssierAbrahams2020}
M.~A.~Nethercote, R.~C.~Assier, and I.~D.~Abrahams,
High-contrast approximation for penetrable wedge diffraction,
\emph{IMA Journal of Applied Mathematics} \textbf{85}(3) (2020), 421--466.
DOI: 10.1093/imamat/hxaa011.

\bibitem{GrothHewettLangdon2018}
S.~P.~Groth, D.~P.~Hewett, and S.~Langdon,
A hybrid numerical-asymptotic boundary element method for high frequency scattering by penetrable convex polygons,
\emph{Wave Motion} \textbf{78} (2018), 32--53.
DOI: 10.1016/j.wavemoti.2017.12.008.

\bibitem{DLMF}
F.~W.~J.~Olver, A.~B.~Olde Daalhuis, D.~W.~Lozier, B.~I.~Schneider, R.~F.~Boisvert, and C.~W.~Clark (eds.),
\emph{NIST Digital Library of Mathematical Functions},
National Institute of Standards and Technology, 2010-- (release updates ongoing).
\url{https://dlmf.nist.gov/}.

\bibitem{WhittakerWatson1927}
E.~T.~Whittaker and G.~N.~Watson,
\emph{A Course of Modern Analysis},
4th ed., Cambridge University Press, Cambridge, 1927.

\bibitem{Lawden1989}
D.~F.~Lawden,
\emph{Elliptic Functions and Applications},
Springer--Verlag, New York, 1989.

\bibitem{BleisteinHandelsman1986}
N.~Bleistein and R.~A.~Handelsman,
\emph{Asymptotic Expansions of Integrals},
Dover Publications, New York, 1986.

\bibitem{Wong2001}
R.~Wong,
\emph{Asymptotic Approximations of Integrals},
SIAM, Philadelphia, 2001.

\bibitem{Matuzas2601}
J.~Matuzas,
Diffraction by a Right-Angle Penetrable Wedge: Closed-Form Solution for \(\nu=\sqrt{2}\),
arXiv:2601.04183, 2026.

\end{thebibliography}
\end{document}